\documentclass[11pt]{article}
\pdfoutput=1

\usepackage[margin=1in,a4paper]{geometry}
\usepackage{amsthm,amssymb,amsmath}  
\usepackage{xspace,enumerate}
\usepackage[dvipsnames]{xcolor}
\usepackage[colorlinks=true,urlcolor=Blue,citecolor=Green,linkcolor=BrickRed]{hyperref}
\usepackage[capitalise]{cleveref}
\usepackage[utf8]{inputenc}
\usepackage{thmtools}
\usepackage{thm-restate}
\usepackage{authblk}
 
\theoremstyle{plain}
\newtheorem{theorem}{Theorem}
\newtheorem{lemma}{Lemma}  
\newtheorem{corollary}[theorem]{Corollary}  
\newtheorem{fact}{Fact}

\theoremstyle{definition}
\newtheorem{definition}{Definition}

\newtheorem{example}{Example}
\newtheorem{remark}[definition]{Remark}
\newtheorem{conjecture}{Conjecture}

\title{Elastic-Degenerate String Matching\\ via Fast Matrix Multiplication}
\author[1]{Giulia Bernardini}
\author[2]{Paweł Gawrychowski}
\author[3]{Nadia Pisanti}
\author[4]{Solon P. Pissis}
\author[5]{Giovanna Rosone}

\affil[1]{Department of Informatics, Systems and Communication (DISCo), University of Milan - Bicocca, Italy\\
\texttt{giulia.bernardini@unimib.it}}
\affil[2]{Institute of Computer Science, University of Wroc\l{}aw, Poland\\
\texttt{gawry@cs.uni.wroc.pl}}
\affil[3]{Department of Computer Science, University of Pisa, Italy and ERABLE Team, INRIA, France\\
\texttt{pisanti@di.unipi.it}}
\affil[4]{CWI, Amsterdam, The Netherlands\\ 
\texttt{solon.pissis@cwi.nl}}
\affil[5]{Department of Computer Science, University of Pisa, Italy\\
\texttt{giovanna.rosone@unipi.it}}

\usepackage{microtype}
\usepackage{amsmath,amsfonts}
\usepackage[cmbtt]{bold-extra}
\usepackage[T1]{fontenc}
\usepackage[utf8]{inputenc}
\usepackage{thm-restate}
\usepackage{comment}
\usepackage{forest}
\usepackage{xspace}
\usepackage{enumerate}
\usepackage{makecell}
\usepackage{listings}
\usepackage{multirow}
\usepackage{diagbox}
\usepackage{todonotes}
\usepackage{thmtools}
\usepackage[ruled,noline,noend]{algorithm2e}
\usepackage[capitalise]{cleveref}
\usepackage{graphics,adjustbox}
\usepackage{tabularx}
\usepackage{amssymb}
\usepackage{epsfig}
\usepackage{verbatim}
\usepackage{braket}
\usepackage[noadjust]{cite}
\usepackage{xspace}
\usepackage{bold-extra}
\usepackage[margin=1in]{geometry}

\newcommand{\rvline}{\hspace*{-2\arraycolsep}\vline\hspace*{-\arraycolsep}}
\newcommand*{\xdash}[1][3em]{\rule[0.5ex]{#1}{0.55pt}}

\newcommand{\N}{\mathcal{N}}

\newcommand{\cO}{\mathcal{O}}
\newcommand{\ctO}{\tilde{\mathcal{O}}}
  \def\dd{\mathinner{.\,.}}
   \newcommand{\defproblem}[3]{
  \vspace{2mm}
\noindent\fbox{
  \begin{minipage}{0.96\textwidth}
  #1\\
  {\bf{INPUT:}} #2  \\
  {\bf{OUTPUT:}} #3
  \end{minipage}
  }
  \vspace{2mm}
}
\DeclareMathOperator{\per}{per}
\DeclareMathOperator{\nodestring}{\mathcal{L}}
\DeclareMathOperator{\wordroot}{root}

\begin{document}
\date{}
\maketitle

\begin{abstract}
An elastic-degenerate (ED) string is a sequence of $n$ sets of strings of total length $N$, which was recently proposed to model a set of similar sequences. The ED string matching (EDSM) problem is to find all occurrences of a pattern of length $m$ in an ED text. The EDSM problem has recently received some attention in the combinatorial pattern matching community, and an $\cO(nm^{1.5}\sqrt{\log m} + N)$-time algorithm is known [Aoyama et al., CPM 2018].
The standard assumption in the prior work on this question is that $N$ is substantially larger than both $n$ and $m$, and thus we would like to have a linear dependency on the former. Under this assumption, the natural open problem is whether we can decrease the 1.5 exponent in the time complexity,
similarly as in the related (but, to the best of our knowledge, not equivalent) {\em word break} problem [Backurs and Indyk, FOCS 2016].

Our starting point is a conditional lower bound for the EDSM problem. We use the popular combinatorial Boolean Matrix Multiplication (BMM) conjecture stating that there is no truly subcubic \emph{combinatorial} algorithm for BMM [Abboud and Williams, FOCS 2014].
By designing an appropriate reduction we show that a combinatorial algorithm solving the EDSM problem in $\cO(nm^{1.5-\epsilon} + N)$ time,
for any $\epsilon>0$, refutes this conjecture. 
Our reduction should be understood as an indication that decreasing the exponent requires fast matrix multiplication.

String periodicity and fast Fourier transform are two standard tools in string algorithms. Our main technical contribution\footnote{A preliminary version of this work appeared in ICALP 2019~\cite{bernardini_et_al:LIPIcs:2019:10597}.} is that we successfully combine these tools with fast matrix multiplication to design 
a non-combinatorial $\ctO(nm^{\omega-1}+N)$-time algorithm for EDSM, where $\omega$ denotes the matrix multiplication exponent and the $\ctO(\cdot)$ notation suppresses polylog factors. To the best of our knowledge, we are the first to combine these tools. 
In particular, using the fact that $\omega<2.373$ [Le Gall, ISSAC 2014; Williams, STOC 2012], we obtain an
$\cO(nm^{1.373} + N)$-time algorithm for EDSM. 
An important building block in our solution, that might find applications in other problems, is a method of selecting a small set of length-$\ell$ substrings of the pattern, called anchors, so that any occurrence of a string from an ED text set contains at least one but not too many such anchors inside. 
\end{abstract}

\thispagestyle{empty}
\clearpage
\setcounter{page}{1}

\section{Introduction}\label{sec:intro}

Boolean matrix multiplication (BMM) is one of the most fundamental computational problems. Apart from its theoretical interest, it has a wide range of applications~\cite{Itai:1977:FMC:800105.803390,Valiant:1975:GCR:1739932.1740048,Furman,Fischer:1971:BMM:1446293.1446319,Munro:1971:EDT:2598952.2599354}. BMM is also the core combinatorial part of integer matrix multiplication. In both problems, we are given two $\N \!\times\! \N$ matrices and we are to compute $\N^2$ values. Integer matrix multiplication can be performed in {\em truly subcubic} time, i.e., in $\cO(\N^{3-\epsilon})$ operations over the field, for some $\epsilon\!>\!0$. The fastest known algorithms for this problem run in $\cO(\N^{2.373})$ time~\cite{DBLP:conf/issac/Gall14a,DBLP:conf/stoc/Williams12}. These algorithms are known as algebraic: they rely on the ring structure of matrices over the field.

There also exists a different family of algorithms for the BMM problem known as  combinatorial. Their focus is on unveiling the combinatorial structure in the Boolean matrices to reduce redundant computations. A series of results~\cite{Arlazaroff:Economical,Bansal:2009:RLC:1747597.1748007,Chan:2015:SUF:2722129.2722145} culminating in an 
$\hat{\cO}(\N^3/\log^{4}\N)$-time algorithm~\cite{YU2018240,Yuicalp} (the $\hat{\cO}(\cdot)$ notation suppresses polyloglog factors) has led to the popular combinatorial BMM conjecture stating that there is no combinatorial algorithm for BMM working in time $\cO(\N^{3-\epsilon})$, for any $\epsilon\!>\!0$~\cite{DBLP:conf/focs/AbboudW14}.
There has been ample work on applying this conjecture to obtain BMM hardness results: see , e.g.,~\cite{Lee:2002:FCG:505241.505242,DBLP:conf/focs/AbboudW14,
DBLP:conf/esa/RodittyZ04,
Henzinger:2015,
DBLP:journals/tcs/LarsenMNT15,
kopelowitz_et_al:2016,
Chatterjee:2017}.

String matching is another fundamental problem, asking to
find all fragments of a string text of length $n$ that match a string pattern of length $m$. This problem has several linear-time solutions~\cite{DBLP:books/daglib/0020103}. In many real-world applications, it is often the case that letters at some positions are either unknown or uncertain. A way of representing these positions is with a
subset of the alphabet $\Sigma$. Such a representation is called {\em degenerate string}. A special case of a degenerate string is when at such unknown or uncertain positions the only subset of the alphabet allowed is the whole alphabet. These special degenerate strings are more commonly known as strings with wildcards. The first efficient algorithm for a text and a pattern, where both may contain wildcards, was published by Fischer and Paterson in 1974~\cite{Fischer-Paterson}. It has undergone several improvements since then~\cite{DBLP:conf/focs/Indyk98a,DBLP:conf/soda/Kalai02a,DBLP:conf/stoc/ColeH02,DBLP:journals/ipl/CliffordC07}. The first efficient algorithm for a standard text and a degenerate pattern, which may contain any non-empty subset of the alphabet, was published by Abrahamson in 1987~\cite{Abrahamson:1987:GSM:37185.37191}, followed by several practically efficient algorithms~\cite{agrep,DBLP:journals/spe/Navarro01,DBLP:journals/jda/HolubSW08}.

 Degenerate letters are used in the IUPAC notation~\cite{IUPAC} to represent a position in a DNA sequence that can have multiple possible alternatives. These are used to encode the consensus of a population of sequences~\cite{PanGenomeConsortium18,wabi18,IndetCPM20} in a multiple sequence alignment (MSA). In the presence of insertions or deletions in the MSA, we may need to consider alternative representations. Consider the following MSA of three closely-related sequences (on the left):
 
     \begin{minipage}{0.25\textwidth}
\begin{enumerate}
	\item[] \texttt{GCA{\color{red}A}C{\color{red}G}GG{\color{red}TA{-}{-}}TT}
	\item[] \texttt{GCA{\color{red}A}C{\color{red}G}GG{\color{red}TATA}TT}
	\item[] \texttt{GCA{\color{red}C}C{\color{red}T}GG{\color{red}{-}{-}{-}{-}}TT}
\end{enumerate}
    \end{minipage}%
    \begin{minipage}{0.7\textwidth}
        \centering
\[ 
\tilde{T}= \left  \{ 
  \begin{tabular}{c}
  \texttt{GCA} 
  \end{tabular}
\right \}
\cdot 
\left  \{ {\color{red}
  \begin{tabular}{c}
  \texttt{A} \\
  \texttt{C} 
  \end{tabular}}
\right \}
\cdot 
\left  \{ 
  \begin{tabular}{c}
  \texttt{C} 
  \end{tabular}
\right \}
\cdot 
\left  \{ {\color{red}
  \begin{tabular}{c}
  \texttt{G}\\
  \texttt{T}
  \end{tabular}}
\right \}
\cdot 
\left  \{ 
  \begin{tabular}{c}
  \texttt{GG} 
  \end{tabular}
\right \}
\cdot 
\left  \{ {\color{red}
  \begin{tabular}{c}
  \texttt{TA}\\ 
  \texttt{TATA}\\
  \texttt{$\varepsilon$} 
  \end{tabular}}
\right \}
\cdot 
\left  \{ 
  \begin{tabular}{c}
  \texttt{TT}\\
  \end{tabular}
\right \}
\]
    \end{minipage}

\noindent These sequences can be compacted into a single sequence $\tilde{T}$ of sets of strings (on the right) containing some deterministic and some non-deterministic
segments. 
A non-deterministic segment is a finite set of deterministic strings and may contain the empty string $\varepsilon$ corresponding to a deletion. The total number of segments is the {\em length} of $\tilde{T}$ and the total number of letters is the {\em size} of $\tilde{T}$. We denote the length by $n=|\tilde{T}|$ and the size by $N=||\tilde{T}||$.

This representation has been defined in~\cite{DBLP:conf/lata/IliopoulosKP17} by Iliopoulos et al.~as an {\em elastic-degenerate} (ED) string. Being a sequence of subsets of $\Sigma^*$, it can be seen as a generalization of a degenerate string. The natural problem that arises is finding all matches of a deterministic pattern $P$ in an ED text $\tilde{T}$. This is the {\em elastic-degenerate string matching} (EDSM) problem. Since its introduction in 2017~\cite{DBLP:conf/lata/IliopoulosKP17}, it has attracted some attention in the combinatorial pattern matching community, and a series of results have been published. The simple algorithm by Iliopoulos et al.~\cite{DBLP:conf/lata/IliopoulosKP17} for EDSM was first improved by Grossi et al.~in the same year, who showed that, for a pattern of length $m$, the EDSM problem can be solved {\em on-line} in $\cO(nm^2+N)$ time~\cite{DBLP:conf/cpm/GrossiILPPRRVV17}; on-line means that it reads the text segment-by-segment and reports an occurrence as soon as this is detected.
This result was improved by Aoyama et al.~\cite{DBLP:conf/cpm/AoyamaNIIBT18} who presented an $\cO(nm^{1.5}\sqrt{\log m}+N)$-time algorithm. An important feature of these bounds is their {\em linear dependency} on $N$. A different branch of on-line algorithms waiving the linear-dependency restriction exists~\cite{DBLP:conf/cpm/GrossiILPPRRVV17,DBLP:conf/wea/PissisR18,sopang,sopang2}. Moreover, the EDSM problem has been considered under Hamming and edit distance~\cite{tcs-ed2020}. 
Recent results on founder block graphs~\cite{makinen2020linear} can also be
casted on elastic-degenerate strings.

A question with a somewhat similar flavor is the {\em word break} problem. 
We are given a dictionary ${\cal D}$, $m=||{\cal D}||$, and a string $S$, $n=|S|$, and the question is whether we can split $S$ into fragments that appear in ${\cal D}$ (the same element of ${\cal D}$ can be used multiple times). Backurs and Indyk~\cite{BackursI16} designed an $\tilde\cO(nm^{1/2-1/18}+m)$-time algorithm for this problem\footnote{ The $\tilde{\cO}(\cdot)$ notation suppresses polylog factors.}. 
Bringmann et al.~\cite{BringmannGL17} improved this to $\tilde\cO(nm^{1/3}+m)$ and showed that this is optimal for combinatorial algorithms by a reduction from $k$-Clique. Their algorithm uses fast Fourier transform (FFT),
and so it is not clear whether it should be considered combinatorial. While this problem seems similar to EDSM, there does not seem to be a direct reduction and so their lower bound does not immediately apply.

\medskip
\noindent\textbf{Our Results.} It is known that BMM and triangle detection (TD) in graphs either both have truly subcubic combinatorial algorithms or none of them do~\cite{williams2010subcubic}. Recall also that the currently fastest algorithm with linear dependency on $N$ for the EDSM problem runs in $\cO(nm^{1.5}\sqrt{\log m}+N)$ time~\cite{DBLP:conf/cpm/AoyamaNIIBT18}. In this paper we prove the following two theorems.

\begin{restatable}{theorem}{thmhardness}\label{thm:hardness}
If the EDSM problem can be solved in $\cO(nm^{1.5-\epsilon} + N)$ time, for any $\epsilon>0$, with a combinatorial algorithm, then there exists a truly subcubic combinatorial algorithm for TD.
\end{restatable}

Arguably, the notion of combinatorial algorithms is not clearly defined, and Theorem~\ref{thm:hardness} should be understood as an indication
that in order to achieve a better complexity one should use fast matrix multiplication.
Indeed, there are examples where a lower bound conditioned on BMM was helpful in constructing efficient algorithms using fast matrix multiplication~\cite{AbboudBW15a, Chang16, BringmannGSW16, Matousek91, CzumajL09, VassilevskaW06, Zwick02}. We successfully design such a non-combinatorial algorithm by combining three ingredients: a string periodicity argument, FFT, and fast matrix multiplication. While periodicity is the usual tool in combinatorial pattern matching~\cite{DBLP:journals/siamcomp/KnuthMP77,DBLP:journals/jacm/CrochemoreP91,DBLP:conf/soda/KociumakaRRW15} and using FFT is also not unusual (for example, it often shows up in approximate string matching~\cite{Abrahamson:1987:GSM:37185.37191,DBLP:journals/jal/AmirLP04,DBLP:journals/ipl/CliffordC07,GawrychowskiU18}), to the best of our knowledge, we are the first to combine these with fast matrix multiplication. Specifically, we show the following result for the EDSM problem, where $\omega$ denotes the matrix multiplication exponent.

\begin{restatable}{theorem}{thmalgo}\label{thm:algo}
The EDSM problem can be solved on-line in $\ctO(nm^{\omega-1}+N)$ time.
\end{restatable}

In order to obtain a faster algorithm for the EDSM problem, we focus on the {\em active prefixes} (AP) problem that lies at the heart of all current solutions~\cite{DBLP:conf/cpm/GrossiILPPRRVV17,DBLP:conf/cpm/AoyamaNIIBT18}. In the AP problem, we are given a string $P$ of length $m$ and a set of arbitrary prefixes $P[1\dd i]$ of $P$, called {\em active prefixes}, stored in a bit vector $U$ so that $U[i]=1$ if $P[1\dd i]$ is active.
We are further given a set ${\cal S}$ of strings of total length $N$ and we are asked to compute a bit vector $V$ which stores the new set of active prefixes of $P$. A new active prefix of $P$ is a concatenation of $P[1\dd i]$ (such that $U[i]=1$) and some element of ${\cal S}$.

Using the algorithmic framework introduced in~\cite{DBLP:conf/cpm/GrossiILPPRRVV17}, EDSM is addressed by solving an instance of the AP problem per each segment $i$ of the ED text corresponding to set ${\cal S}$ of the AP problem.
Hence, an $\cO(f(m)+N_i)$ solution for the AP problem (with $N_i$ being the size of a single segment of the ED text) implies an $\cO(nf(m)+N)$ solution of EDSM, as $f(m)$ is repeated $n$ times and $N = \sum_{i=1}^n N_i$. 
The algorithm of~\cite{DBLP:conf/cpm/AoyamaNIIBT18} solves the AP problem in $\cO(m^{1.5}\sqrt{\log m} + N_i)$ time leading to $\cO(nm^{1.5}\sqrt{\log m} + N)$ time for the EDSM problem. 
Our algorithm partitions the strings of each segment $i$ of the ED text into three types according to a periodicity criterion, and then solves a restricted instance of the AP problem for each of the types. In particular, we solve the AP problem in $\ctO(m^{\omega-1}+N_i)$ time leading to $\ctO(nm^{\omega-1}+N)$ time for the EDSM problem.
Given this connection between the two problems and, in particular, between their size parameter $N$, in the rest of the paper we will denote with $N$ also the parameter $N_i$ of the AP problem. 

An important building block in our solution that might find applications in other problems is a method of selecting a small set of length-$\ell$ substrings of the pattern, called \emph{anchors}, so that any relevant occurrence of a string from an ED text set contains at least one but not too many such anchors inside. This is obtained by rephrasing the question in a graph-theoretical language and then
generalizing the well-known fact that an instance of the hitting set problem with $m$ sets over $[n]$, each of size at least $k$, has a solution
of size $\cO(n/k\cdot \log m)$.
While the idea of carefully selecting some substrings of the same length is not new (for example Kociumaka et al.~\cite{DBLP:conf/soda/KociumakaRRW15}
used it to design a data structure for pattern matching queries on a string), our setting is different and hence so is the method
of selecting these substrings. 

In addition to the conditional lower bound for the EDSM problem (Theorem~\ref{thm:hardness}), which also appeared in~\cite{bernardini_et_al:LIPIcs:2019:10597}, we also show here the following conditional lower bound for the AP problem.

\begin{restatable}{theorem}{APthmhardness}\label{the:AP}
If the AP problem can be solved in $\cO(m^{1.5-\epsilon}+N)$ time, for any $\epsilon>0$, with a combinatorial algorithm, then there exists a truly subcubic combinatorial algorithm for the BMM problem.
\end{restatable}

\medskip
\noindent\textbf{Roadmap.}
Section~\ref{sec:prel} provides the necessary definitions and notation as well as the algorithmic toolbox used throughout the paper. In Section~\ref{app:app} we prove our 
lower bound result for the AP problem (Theorem~\ref{the:AP}). The 
lower bound result for the EDSM problem is proved in Section~\ref{sec:lb}  (Theorem~\ref{thm:hardness}). In Section~\ref{sec:algo} we present our algorithm for EDSM (Theorem~\ref{thm:algo}); this is the most technically involved part of the paper.

\section{Preliminaries}\label{sec:prel}

Let $T=T[1]T[2]\ldots T[n]$ be a string of length $|T|=n$ over a finite ordered alphabet $\Sigma$ of size $|\Sigma|=\sigma$. For two positions $i$ and $j$ on $T$, we denote by $T[i\dd j]=T[i]\ldots T[j]$ the substring of $T$ that starts at position $i$ and ends at position $j$ (it is of length $0$ if $j<i$). By $\varepsilon$ we denote the empty string of length 0. A prefix of $T$ is a substring of the form $T[1\dd j]$, and a suffix of $T$ is a substring of the form $T[i\dd n]$. $T^{r}$ denotes the reverse of $T$, that is, $T[n] T[n-1] \ldots T[1]$. We say that a string $X$ is a power of a string $Y$ if there exists an integer $k>1$, such that $X$ is expressed as $k$ consecutive concatenations of $Y$, denoted by $X=Y^k$. A period of a string $X$ is any integer $p\in [1,|X|]$ such that $X[i]=X[i+p]$ for every $i=1,2,\ldots,|X|-p$, and {\em the period}, denoted by $\per(X)$, is the smallest such $p$. We call a string $X$ \emph{strongly periodic} if $\per(X)\leq |X|/4$. 

\begin{lemma}[\cite{PeriodicityLemma}]
\label{lem:periodicity}
If $p$ and $q$ are both periods of the same string $X$, and additionally $p+q\leq |X|+1$, then $\gcd(p,q)$ is also a period of $X$.
\end{lemma}

A \emph{trie} is a tree in which every edge is labeled with a single letter, and every two edges outgoing from the same node have different labels.
The label of a node $u$ in such a tree $T$, denoted by $\nodestring(u)$, is defined as the concatenation of the labels of all the edges on the path from
the root of $T$ to $u$. Thus, the label of the root of $T$ is $\varepsilon$, and a trie is a representation of a set of strings consisting of the labels
of all its leaves. By replacing each path $p$ consisting of nodes with exactly one child by an edge labeled by the concatenation of the
labels of the edges of $p$ we obtain a \emph{compact trie}. The nodes of the trie that are removed after this transformation are called
\emph{implicit}, while the remaining ones are referred to as \emph{explicit}.
The suffix tree of a string $S$ is the compact trie representing all suffixes
of $S\$$, $\$\notin\Sigma$, where instead of explicitly storing the label $S[i\dd j]$ of an edge we represent it by the pair $(i,j)$.

A \emph{heavy path decomposition} of a tree $T$ is obtained by selecting, for every non-leaf node $u\in T$, its child $v$ such that the subtree rooted at $v$ is the largest. This decomposes the nodes of $T$ into node-disjoint paths, with each such path $p$ (called a heavy path) starting at some node, called the \emph{head} of $p$, and ending at a leaf. An important property of such a decomposition is that the number of distinct heavy paths above any leaf (that is, intersecting the path from a leaf to the root) is only logarithmic in the size of $T$~\cite{DBLP:journals/jcss/SleatorT83}.

Let $\tilde{\Sigma}$ denote the set of all finite non-empty subsets of $\Sigma^{*}$. Previous works (cf.~\cite{DBLP:conf/lata/IliopoulosKP17,DBLP:conf/cpm/GrossiILPPRRVV17, DBLP:conf/cpm/AoyamaNIIBT18,DBLP:conf/wea/PissisR18,tcs-ed2020}) define $\tilde{\Sigma}$ as the set of all finite non-empty subsets of $\Sigma^{*}$ excluding $\{\varepsilon\}$ but we waive here the latter restriction as it has no algorithmic implications. An {\em elastic-degenerate string} $\tilde{T}=\tilde{T}[1]\dots\tilde{T}[n]$, or ED string, over alphabet $\Sigma$, is a string over $\tilde{\Sigma}$, i.e., an ED string is an element of $\tilde{\Sigma}^{*}$, and hence each $\tilde{T}[i]$ is a set of strings.

Let $\tilde{T}$ denote an ED string of length $n$, i.e. $|\tilde{T}| = n$. 
We assume that for any $1 \leq i \leq n$, the set 
$\tilde{T}[i]$ $\in \tilde{\Sigma}$
is implemented as an array and can be accessed by an index, i.e., $\tilde{T}[i] =\{\tilde{T}[i][k]~|~k = 1,\ldots,|\tilde{T}[i]|\}$.
For any $\tilde{\sigma}$ $\in \tilde{\Sigma}$, $||\tilde{\sigma}||$ denotes the
total length of all strings in $\tilde{\sigma}$, and for any ED string $\tilde{T}$,  $||\tilde{T}||$ denotes the total length of all strings in all $\tilde{T}[i]$s. We will denote $N_i = \sum_{k=1}^{|\tilde{T}[i]|}|\tilde{T}[i][k]|$ the
total length of all strings in $\tilde{T}[i]$ and $N = \sum_{i=1}^{n} ||\tilde{T}[i]||$ the \emph{size} of $\tilde{T}$. 
An ED string $\tilde{T}$ can be thought of as a compact representation of the set of strings 
${\cal A}(\tilde{T})$ which is the Cartesian product of all $\tilde{T}[i]$s; that is, 
${\cal A}(\tilde{T}) =\tilde{T}[1]\times \ldots \times \tilde{T}[n]$ where $A \times B = \{xy \mid x \in A, y \in B\}$ for any sets of strings $A$ and $B$.\\ 
For any ED string $\tilde{X}$ and a pattern $P$, we say that $P$ {\em matches} $\tilde{X}$ if:
\begin{enumerate}
\item $|\tilde{X}| = 1$ and $P$ is a substring of some string in $\tilde{X}[1]$, or,
\item $|\tilde{X}| > 1$ and $P = P_{1}\ldots P_{|\tilde{X}|}$, where $P_1$ is a suffix of some string in $\tilde{X}[1]$, $P_{|\tilde{X}|}$ is a prefix of some string in $\tilde{X}[|\tilde{X}|]$, and $P_i \in \tilde{X}[i]$, for all $1 < i < |\tilde{X}|$.
\end{enumerate}

We say that an  occurrence of a string $P$ ends at position $j$ of an ED string $\tilde{T}$ if there exists $i\leq j$ such that $P$ matches $\tilde{T}[i]\ldots\tilde{T}[j]$. 
We will refer to string $P$ as the {\em pattern} and to ED string $\tilde{T}$ as the {\em text}. We define the main problem considered in this paper.

{\defproblem{\textsc{Elastic-Degenerate String Matching} (EDSM)}{A string $P$ of length $m$ and an ED string $\tilde{T}$ of length $n$ and size $N \geq m$.}
{All positions in $\tilde{T}$ where at least one occurrence of $P$ ends.}}

\begin{example}
Pattern $P=\textnormal{\texttt{GTAT}}$ ends at positions $2$, $6$, and $7$ of the following text $\tilde{T}$.
\[ \tilde{T} = 
\left  \{ 
  \begin{tabular}{c}
  \textnormal{\texttt{AT{\color{red}GTA}}} 
  \end{tabular}
\right \}
\cdot 
\left  \{ 
  \begin{tabular}{c}
  \textnormal{\texttt{A}} \\
  \textnormal{\texttt{\color{red}{T}}} 
  \end{tabular}
\right \}
\cdot 
\left  \{ 
  \begin{tabular}{c}
  \textnormal{\texttt{C}} 
  \end{tabular}
\right \}
\cdot 
\left  \{ 
  \begin{tabular}{c}
  \textnormal{\texttt{G}}\\
  \textnormal{\texttt{T}}
  \end{tabular}
\right \}
\cdot 
\left  \{ 
  \begin{tabular}{c}
  \textnormal{\texttt{C{\color{red}{G}}}} 
  \end{tabular}
\right \}
\cdot 
\left  \{ 
  \begin{tabular}{c}
  \textnormal{\texttt{{\color{red}{TA}}}}\\ 
  \textnormal{\texttt{{\color{red}{TAT}}A}}\\
  \textnormal{\texttt{\color{red}{$\varepsilon$}}}
  \end{tabular}
\right \}
\cdot 
\left  \{ 
  \begin{tabular}{c}
  \textnormal{\texttt{{\color{red}{TAT}}GC}}\\
  \textnormal{\texttt{{\color{red}{T}}TTTA}}
  \end{tabular}
\right \}
\]
\end{example}

Aoyama et al.~\cite{DBLP:conf/cpm/AoyamaNIIBT18} obtained an on-line $\cO(nm^{1.5}\sqrt{\log m}+N)$-time algorithm
by designing an efficient solution for the following problem.

{\defproblem{\textsc{Active Prefixes} (AP)}{A string $P$ of length $m$, a bit vector $U$ of size $m$, a set ${\cal S}$ of strings of total length $N$.} {A bit vector $V$ of size $m$ with $V[j]=1$ if and only if there exists $S \in {\cal S}$ and $i \in[1,m], U[i]=1$, such that $P[1 \dd i]\cdot S=P[1\dd i+|S|]$ and $j = i+|S|$.}}

In more detail, given an ED text $\tilde{T}=\tilde{T}[1]\dots \tilde{T}[n]$, one should consider an instance of the AP problem per each $\tilde{T}[i]$. Hence, an $\cO(f(m)+N_i)$ solution for AP ( 
$N_i$ being the size of $\tilde{T}[i]$) 
implies an $\cO(n\cdot f(m)+N)$ solution for EDSM, as $f(m)$ is repeated $n$ times and $N = \sum_{i=1}^n N_i$. We provide an example of the AP problem.

\begin{example}
Let $P=\textnormal{\texttt{ababbababab}}$ of length $m=11$, $U=01000100000$, and ${\cal S}=\{\boldsymbol{\varepsilon}, \textnormal{\texttt{ab}},\textnormal{\texttt{abb}}, \textnormal{\texttt{ba}}, \textnormal{\texttt{baba}}\}$. We have that $V=01011101010$.
\end{example}

For our lower bound results we rely on BMM and the following closely related problem.

{\defproblem{\textsc{Boolean Matrix Multiplication} (BMM)}{ Two $\N \times \N$ Boolean matrices $A$ and $B$.}
{$\N \times \N$ Boolean matrix $C$, where $C[i,j] =  {\bigvee\limits_{k}} (A[i,k] \wedge B[k,j])$.}}

{\defproblem{\textsc{Triangle Detection} (TD)}{Three $\N \times \N$ Boolean matrices $A,B$ and $C$.}
{Are there $i,j,k$ such that $A[i,j]=B[j,k]=C[k,i]=1$?}}

An algorithm is called \emph{truly subcubic} if it runs in
$\cO(\N^{3-\epsilon})$ time, for some $\epsilon > 0$. TD and BMM either
both have truly subcubic combinatorial algorithms, or none of them do~\cite{williams2010subcubic}.

\section{AP Conditional Lower Bound}\label{app:app}

To investigate the hardness of the EDSM problem, we first show that an $\cO(m^{1.5-\epsilon}+N)$-time solution to the active prefixes problem, that constitutes the core of the solutions proposed in~\cite{DBLP:conf/cpm/GrossiILPPRRVV17,DBLP:conf/cpm/AoyamaNIIBT18}, would imply a truly subcubic combinatorial algorithm for Boolean matrix multiplication (BMM). We recall that in the AP problem, we are given a string $P$ of length $m$ and a set of prefixes $P[1\dd i]$ of $P$, called {\em active prefixes}, stored in a bit vector $U$ ($U[i]=1$ if and only if $P[1\dd i]$ is active). We are further given a set ${\cal S}$ of strings of total length $N$ and we are asked to compute a bit vector $V$ storing the new set of active prefixes of $P$: a prefix of $P$ that extends $P[1\dd i]$ (such that $U[i]=1$) with some element of ${\cal S}$. 
Of course, we can solve BMM by working over integers and using one of the
fast matrix multiplication algorithms; plugging in the best known bounds results in an $\cO(\N^{2.373})$-time
algorithm~\cite{DBLP:conf/issac/Gall14a,DBLP:conf/stoc/Williams12}. However, such an algorithm is not
\emph{combinatorial}, i.e., it uses \emph{algebraic} methods. In comparison, the best known combinatorial
algorithm for BMM works in $\hat\cO(\N^3/\log^{4}\N)$ time~\cite{YU2018240}.
This leads to the following popular conjecture.

\begin{conjecture}[\cite{DBLP:conf/focs/AbboudW14}]\label{conj:BMM}
There is no combinatorial algorithm for the BMM problem working in time $\cO(\N^{3-\epsilon})$, for any $\epsilon>0$.
\end{conjecture}

Aoyama et al.~\cite{DBLP:conf/cpm/AoyamaNIIBT18} showed that the AP problem
can be solved in $\cO(m^{1.5}\sqrt{\log m}+N)$ time for constant-sized alphabets. Together with some
standard string-processing techniques applied similarly as in~\cite{DBLP:conf/cpm/GrossiILPPRRVV17},
this is then used to solve the EDSM problem by creating an instance of the AP problem for every set $\tilde{T}[i]$ of $\tilde{T}$, i.e., with ${\cal S}=\tilde{T}[i]$. 

We argue that, unless Conjecture~\ref{conj:BMM} is false, the AP problem cannot be solved faster than $\cO(m^{1.5-\epsilon}+N)$, for any $\epsilon>0$, with a combinatorial algorithm (note that the algorithm of Aoyama et al.~\cite{DBLP:conf/cpm/AoyamaNIIBT18} uses FFT, and so it is not completely clear
whether it should be considered to be combinatorial).
We show this by a reduction from combinatorial BMM. Assume that, for the AP problem, we seek combinatorial algorithms with the running time $\cO(m^{1.5-\epsilon}+N)$, i.e., with linear dependency on the total
length of the strings. We need to show that such an algorithm implies that the BMM problem can be solved in $\cO(\N^{3-\epsilon'})$ time, for some $\epsilon'>0$, with a combinatorial algorithm, thus implying that Conjecture~\ref{conj:BMM} is false.

\APthmhardness*

\begin{proof}
Recall that in the BMM problem the matrices are denoted by $A$ and $B$. 
In order to compute $C\!=\!A\! \times\! B$, we need to find, for
every $i,j=1,\ldots,\N$, an index $k$ such that $A[i,k]=1$ and $B[k,j]=1$. To this purpose, we split matrix $A$ into
blocks of size $\N\! \cdot \!L$ and $B$ into blocks of size  $L\!\cdot\! L$. This corresponds to
considering values of $j$ and $k$ in intervals of size $L$, and clearly there are $\N/L$ such
intervals. Matrix $B$ is thus split into $(\N/L)^2$ blocks, giving rise to an equal number of instances of the AP problem, each one corresponding to an interval of $j$ and an interval of $k$.
This creates $(\N/L)^2$ blocks in matrix $B$; we will thus create $(\N/L)^2$ separate instances of the AP problem corresponding to an interval of $j$ and an interval of $k$. 
We will now describe the instance corresponding to the $(K,J)$-th block, where $1 \leq K,J \leq \N/L$.

We build the string $P$ of the AP problem, for any block, as a concatenation of $\N$ gadgets corresponding to $i=1,\ldots,\N$, and we construct
the bit vector $U^{(K,J)}$ of the AP problem as a concatenation of $\N$ bit vectors, one per gadget. Each gadget consists of  the same string $\texttt{a}^L \texttt{b} \texttt{a}^L$; we set to $1$ the $k'$-th bit of the $i$-th gadget bit vector if $A[i,(K-1)L+k']=1$. 
The solution of the AP problem $V^{(K,J)}$
will allow us to recover the solution of BMM, as we will ensure that the bit corresponding to the $j'$-th $\texttt{a}$ in the second half of the gadget is set to $1$ if and only if, for some $k'\in[L]$, $A[i,(K-1)L+k']=1$ and
$B[(K-1)L+k',(J-1)L+j']=1$. 
In order to enforce this, we will include the following strings in set ${\cal S}^{(K,J)}$:
$$\texttt{a}^{L-k'} \texttt{b} \texttt{a}^{j'}, \text{ for every } k',j'\in[L] \text{ such that } B[(K-1)L+k',(J-1)L+j']=1.$$ 
This guarantees that after solving the AP problem we have the required property, and thus, after solving all the instances, we have obtained matrix $C\!=\!A \!\times\! B$.
Indeed, consider values $j$, i.e., the index that runs on the columns of $C$, in intervals of size $L$. By construction and by definition of BMM, the $i$-th line of the $J$-th column interval of $C$ is obtained by taking the disjunction of the second half of the $i$-th interval of each $(K,J)$-th bit vector
for every $K=1,2,\dots,\N/L$.

We have a total of $(\N/L)^2$ instances. In each of them, the total length of all strings is $\cO(L^3)$, and the length of the input string $P$ is $(2L+1)\N=\cO(L\cdot \N)$. Using our assumed algorithm for each instance, we obtain the following total time:
$$\cO((\N/L)^2 \cdot(L^3 + (\N\cdot L)^{1.5-\epsilon})) = \cO(\N^2\cdot L + \N^{3.5-\epsilon}/L^{0.5+\epsilon}).$$
If we set $L=\N^{(1.5-\epsilon)/(1.5+\epsilon)}$, then the total time becomes:
\begin{eqnarray*}
&&\cO(\N^{2+(1.5-\epsilon)/(1.5+\epsilon)} + \N^{3.5-\epsilon-(0.5+\epsilon)(1.5-\epsilon)/(1.5+\epsilon)}) \\
&=& \cO(\N^{2+(1.5-\epsilon)/(1.5+\epsilon)} + \N^{2+(1.5-\epsilon)-(1.5-\epsilon)(0.5+\epsilon)/(1.5+\epsilon)})\\
&=& \cO(\N^{2+(1.5-\epsilon)/(1.5+\epsilon)} + \N^{2+(1.5-\epsilon)(1.5+\epsilon-0.5-\epsilon)/(1.5+\epsilon)})\\
&=& \cO(\N^{2+(1.5-\epsilon)/(1.5+\epsilon)}).
\end{eqnarray*}
Hence we obtain a combinatorial BMM algorithm with complexity $\cO(\N^{3-\epsilon'})$ , where $\epsilon' = 1-(1.5-\epsilon)/(1.5+\epsilon) > 0$. 
\end{proof}

\begin{example}
Consider the following instance of the BMM problem with $\N=6$ and $L=3$. 
\small
\begin{align*}
&\hspace{48pt}A & &\hspace{48pt}B & &\hspace{48pt}C\\
&\begin{bmatrix}
\begin{matrix}
  0 & 1 & 0  \\
  1 & 0 & 1 \\
  0 & 0 & 0
  \end{matrix}
  & \hspace{-8pt} \rvline \hspace{-8pt}&
  \hspace{-8pt}\begin{matrix}
  0 & 1 & 0 \\
  0 & 0 & 0 \\
  0 & 0 & 1
  \end{matrix}\\
   \hspace{-5pt}\xdash[3.5em] & & \hspace{-13pt}\xdash[3.5em]\\
 \begin{matrix}
  1 & 0 & 0  \\
  0 & 0 & 0 \\
  0 & 1 & 0 
  \end{matrix}
  & \hspace{-8pt} \rvline &
  \hspace{-8pt}\begin{matrix}
  0 & 1 & 0 \\
  1 & 0 & 0 \\
  0 & 0 & 0
  \end{matrix}
\end{bmatrix} \hspace{22pt}\times \hspace{-8pt}&
&\begin{bmatrix}
\begin{matrix}
  0 & 0 & 0  \\
  1 & 0 & 0 \\
  0 & 0 & 1
  \end{matrix}
  & \hspace{-8pt} \rvline &
  \hspace{-8pt}\begin{matrix}
  0 & 0 & 1 \\
  0 & 0 & 0 \\
  0 & 1 & 0
  \end{matrix}\\
 \hspace{-5pt}\xdash[3.5em] & & \hspace{-13pt}\xdash[3.5em]\\
 \begin{matrix}
  0 & 1 & 0  \\
  0 & 0 & 0 \\
  1 & 0 & 0 
  \end{matrix}
  & \hspace{-8pt} \rvline &
  \hspace{-8pt}\begin{matrix}
  0 & 0 & 0 \\
  1 & 0 & 0 \\
  0 & 1 & 0
  \end{matrix}
\end{bmatrix} \hspace{22pt}= \hspace{-8pt}&
&\begin{bmatrix}
\begin{matrix}
  \mathbf{1} & \mathbf{0} & \mathbf{0}  \\
  0 & 0 & 1 \\
  1 & 0 & 0
  \end{matrix}
  & \hspace{-8pt} \rvline &
  \hspace{-8pt}\begin{matrix}
  1 & 0 & 0 \\
  0 & 1 & 1 \\
  0 & 1 & 0
  \end{matrix}\\
 \hspace{-5pt}\xdash[3.5em] & & \hspace{-13pt}\xdash[3.5em]\\
 \begin{matrix}
  0 & 0 & 0  \\
  0 & 1 & 0 \\
  1 & 0 & 0 
  \end{matrix}
 & \hspace{-8pt} \rvline &
  \hspace{-8pt}\begin{matrix}
  1 & 0 & 1 \\
  0 & 0 & 0 \\
  0 & 0 & 0
  \end{matrix}
\end{bmatrix}
\end{align*}
\\

\normalsize
From matrices $A$ and $B$, we now show how the resulting matrix $C$ can be found by building and solving $4$ instances of the AP problem constructed as follows. The pattern is 

\[P=\textnormal{\texttt{\normalsize aaabaaa}}\cdot \textnormal{\texttt{\normalsize aaabaaa}}\cdot \textnormal{\texttt{\normalsize aaabaaa}}\cdot \textnormal{\texttt{\normalsize aaabaaa}}\cdot \textnormal{\texttt{\normalsize aaabaaa}}\cdot \textnormal{\texttt{\normalsize aaabaaa}}\]

where the six gadgets are separated by a $'\cdot'$ to be highlighted. For the AP instances, the vectors
$U^{(K,J)}$ shown below are the input bit vectors, and the sets $S^{(K,J)}$ are the input set of strings.

For each instance, the bit vector $V^{(K,J)}$ shown below is the output of the AP problem.

\begingroup 
\small
\setlength\arraycolsep{1.5pt}
\begin{align*}
&i & &\hspace{23pt} 1 & &\hspace{7pt}2 & &\hspace{7pt}3 & &\hspace{7pt}4 & &\hspace{7pt}5 & &\hspace{7pt}6\\
&U^{(1,1)}: \hspace{-5pt}& [&\begin{matrix}
0 & 1 & 0 & 0 & 0 & 0 & 0
\end{matrix}
\hspace{5pt}
|& & \hspace{-19pt}\begin{matrix}
1 & 0 & 1 & 0 & 0 & 0 & 0
\end{matrix}
\hspace{5pt}
|& & \hspace{-19pt}\begin{matrix}
0 & 0 & 0 & 0 & 0 & 0 & 0
\end{matrix}
\hspace{5pt}
|& & \hspace{-19pt}\begin{matrix}
1 & 0 & 0 & 0 & 0 & 0 & 0
\end{matrix}
\hspace{5pt}
|& & \hspace{-19pt}\begin{matrix}
0 & 0 & 0 & 0 & 0 & 0 & 0
\end{matrix}
\hspace{5pt}
|& & \hspace{-19pt}\begin{matrix}
0 & 1 & 0 & 0 & 0 & 0 & 0
\end{matrix}] \\
&S^{(1,1)}: & &\hspace{-4pt}\{\textnormal{\texttt{\normalsize aba,baaa}}\}\\
&V^{(1,1)}: \hspace{-5pt}& [&\begin{matrix}
0 & 0 & 0 & 0 & \mathbf{1} & \mathbf{0} & \mathbf{0}
\end{matrix}
\hspace{3pt}
|& & \hspace{-19pt}\begin{matrix}
0 & 0 & 0 & 0 & 0 & 0 & 1
\end{matrix}
\hspace{5pt}
|& & \hspace{-19pt}\begin{matrix}
0 & 0 & 0 & 0 & 0 & 0 & 0
\end{matrix}
\hspace{5pt}
|& & \hspace{-19pt}\begin{matrix}
0 & 0 & 0 & 0 & 0 & 0 & 0
\end{matrix}
\hspace{5pt}
|& & \hspace{-19pt}\begin{matrix}
0 & 0 & 0 & 0 & 0 & 0 & 0
\end{matrix}
\hspace{5pt}
|& & \hspace{-19pt}\begin{matrix}
0 & 0 & 0 & 0 & 1 & 0 & 0
\end{matrix}]\\
& \\
&U^{(1,2)}: \hspace{-5pt}& [&\begin{matrix}
0 & 1 & 0 & 0 & 0 & 0 & 0
\end{matrix}
\hspace{5pt}
|& & \hspace{-19pt}\begin{matrix}
1 & 0 & 1 & 0 & 0 & 0 & 0
\end{matrix}
\hspace{5pt}
|& & \hspace{-19pt}\begin{matrix}
0 & 0 & 0 & 0 & 0 & 0 & 0
\end{matrix}
\hspace{5pt}
|& & \hspace{-19pt}\begin{matrix}
1 & 0 & 0 & 0 & 0 & 0 & 0
\end{matrix}
\hspace{5pt}
|& & \hspace{-19pt}\begin{matrix}
0 & 0 & 0 & 0 & 0 & 0 & 0
\end{matrix}
\hspace{5pt}
|& & \hspace{-19pt}\begin{matrix}
0 & 1 & 0 & 0 & 0 & 0 & 0
\end{matrix}]\\
&S^{(1,2)}: & &\hspace{-4pt}\{\textnormal{\texttt{\normalsize aabaaa,baa}}\}\\
&V^{(1,2)}: \hspace{-5pt}& [&\begin{matrix}
0 & 0 & 0 & 0 & 0 & 0 & 0
\end{matrix}
\hspace{5pt}
|& & \hspace{-19pt}\begin{matrix}
0 & 0 & 0 & 0 & 0 & 1 & 1
\end{matrix}
\hspace{5pt}
|& & \hspace{-19pt}\begin{matrix}
0 & 0 & 0 & 0 & 0 & 0 & 0
\end{matrix}
\hspace{5pt}
|& & \hspace{-19pt}\begin{matrix}
0 & 0 & 0 & 0 & 0 & 0 & 1
\end{matrix}
\hspace{5pt}
|& & \hspace{-19pt}\begin{matrix}
0 & 0 & 0 & 0 & 0 & 0 & 0
\end{matrix}
\hspace{5pt}
|& & \hspace{-19pt}\begin{matrix}
0 & 0 & 0 & 0 & 0 & 0 & 0
\end{matrix}]\\
& \\
&U^{(2,1)}: \hspace{-5pt}& [&\begin{matrix}
0 & 1 & 0 & 0 & 0 & 0 & 0
\end{matrix}
\hspace{5pt}
\hspace{1pt}|& & \hspace{-19pt}\begin{matrix}
0 & 0 & 0 & 0 & 0 & 0 & 0
\end{matrix}
\hspace{5pt}
|& & \hspace{-19pt}\begin{matrix}
0 & 0 & 1 & 0 & 0 & 0 & 0
\end{matrix}
\hspace{5pt}
|& & \hspace{-19pt}\begin{matrix}
0 & 1 & 0 & 0 & 0 & 0 & 0
\end{matrix}
\hspace{5pt}
|& & \hspace{-19pt}\begin{matrix}
1 & 0 & 0 & 0 & 0 & 0 & 0
\end{matrix}
\hspace{5pt}
|& & \hspace{-19pt}\begin{matrix}
0 & 0 & 0 & 0 & 0 & 0 & 0
\end{matrix}]\\
&S^{(2,1)}: & &\hspace{-4pt}\{\textnormal{\texttt{\normalsize aabaa,ba}}\}\\
&V^{(2,1)}: \hspace{-5pt}& [&\begin{matrix}
0 & 0 & 0 & 0 & \mathbf{0} & \mathbf{0} & \mathbf{0} 
\end{matrix}
\hspace{3pt}
|& & \hspace{-19pt}
\begin{matrix}
0 & 0 & 0 & 0 & 0 & 0 & 0
\end{matrix}
\hspace{5pt}
|& & \hspace{-19pt}\begin{matrix}
0 & 0 & 0 & 0 & 1 & 0 & 0
\end{matrix}
\hspace{5pt}
|& & \hspace{-19pt}\begin{matrix}
0 & 0 & 0 & 0 & 0 & 0 & 0
\end{matrix}
\hspace{5pt}
|& & \hspace{-19pt}\begin{matrix}
0 & 0 & 0 & 0 & 0 & 1 & 0
\end{matrix}
\hspace{5pt}
|& & \hspace{-19pt}\begin{matrix}
0 & 0 & 0 & 0 & 0 & 0 & 0
\end{matrix}]\\
& \\
&U^{(2,2)}: \hspace{-5pt}& [&\begin{matrix}
0 & 1 & 0 & 0 & 0 & 0 & 0
\end{matrix}
\hspace{5pt}
|& & \hspace{-19pt}\begin{matrix}
0 & 0 & 0 & 0 & 0 & 0 & 0
\end{matrix}
\hspace{5pt}
|& & \hspace{-19pt}\begin{matrix}
0 & 0 & 1 & 0 & 0 & 0 & 0
\end{matrix}
\hspace{5pt}
|& & \hspace{-19pt}\begin{matrix}
0 & 1 & 0 & 0 & 0 & 0 & 0
\end{matrix}
\hspace{5pt}
|& & \hspace{-19pt}\begin{matrix}
1 & 0 & 0 & 0 & 0 & 0 & 0
\end{matrix}
\hspace{5pt}
|& & \hspace{-19pt}\begin{matrix}
0 & 0 & 0 & 0 & 0 & 0 & 0
\end{matrix}]\\
&S^{(2,2)}: & &\hspace{-4pt}\{\textnormal{\texttt{\normalsize aba,baa}}\}\\
&V^{(2,2)}: \hspace{-5pt}& [&\begin{matrix}
0 & 0 & 0 & 0 & 1 & 0 & 0
\end{matrix}
\hspace{5pt}
|& & \hspace{-19pt}\begin{matrix}
0 & 0 & 0 & 0 & 0 & 0 & 0
\end{matrix}
\hspace{5pt}
|& & \hspace{-19pt}\begin{matrix}
0 & 0 & 0 & 0 & 0 & 1 & 0
\end{matrix}
\hspace{5pt}
|& & \hspace{-19pt}\begin{matrix}
0 & 0 & 0 & 0 & 1 & 0 & 0
\end{matrix}
\hspace{5pt}
|& & \hspace{-19pt}\begin{matrix}
0 & 0 & 0 & 0 & 0 & 0 & 0
\end{matrix}
\hspace{5pt}
|& & \hspace{-19pt}\begin{matrix}
0 & 0 & 0 & 0 & 0 & 0 & 0
\end{matrix}]
\end{align*}
\endgroup

As an example on how to obtain matrix $C$, consider the bold part of $C$ above (\emph{i.e.,} the first line of block $(1,1)$ of $C$). This is obtained by taking the disjunction of the bold parts of $V^{(1,1)}$ and $V^{(2,1)}$.
\end{example}
\normalsize

\section{EDSM Conditional Lower Bound}\label{sec:lb}
Since the lower bound for the AP problem does not imply {\em per se} a lower bound for the whole EDSM problem, in this section we show a conditional lower bound for the EDSM problem.
Specifically, 
we perform a reduction from Triangle Detection to show that, if the EDSM problem could be solved in $\cO(nm^{1.5-\epsilon}+N)$ time, this would imply the existence of a truly subcubic algorithm for TD.
We show that TD can be reduced to the  decision version of the
EDSM problem: the goal is to detect whether
there exists at least one occurrence of $P$ in $\tilde{T}$.
To this aim, given three matrices $A$, $B$, $C$, we first decompose matrix $B$ into blocks of size $\mathcal{N}/s\times\mathcal{N}/s$, where $s$ is a parameter to be determined later; the pattern $P$ is obtained by concatenating a number (namely $z=\mathcal{N}s^2$) of constituent parts $P_i$ of length $\cO(\mathcal{N}/s)$, each one built with five letters from disjoint subalphabets. The ED text $\tilde{T}$ is composed of three parts: the central part consists of three degenerate segments, the first one encoding the 1s of matrix $A$, the second one those of matrix $B$ and the third one those of matrix $C$. These segments are built in such a way that the concatenation of strings of subsequent segments is of the same form as the pattern's building blocks. This central part is then padded to the left and to the right with sets containing appropriately chosen concatenations of substrings $P_i$ of $P$, so that an occurrence of the pattern in the text implies that one of its building blocks matches the central part of the text,
thus corresponding to a triangle. Formally:

\thmhardness*
\begin{proof}
Consider an instance of TD, where we are given three $\N\times\N$ Boolean matrices
$A,B,C$, and the question is to check if there exist $i,j,k$ such that $A[i,j]=B[j,k]=C[k,i]=1$.
Let $s$ be a parameter, to be determined later, that corresponds to decomposing $B$ into blocks of size $(\N/s)\times(\N/s)$.
We reduce to an instance of EDSM over an alphabet $\Sigma$ of size $\cO(\N)$.

\noindent\textbf{Pattern $P$.}
We construct $P$ by concatenating, in some fixed order, the following strings:
$$P(i,x,y) = v(i) x a^{\N/s} x \$\$ y a^{\N/s} y v(i)$$
for every $i=1,2,\dots,\N$ and $x,y=1,2,\dots,s$,
where $a \in \Sigma_1$, $\$ \in \Sigma_{2}$, $x \in \Sigma_3$, $y \in \Sigma_4$, $v(i) \in \Sigma_5$,
and $\Sigma_{1},\Sigma_{2},\dots,\Sigma_{5}$ are disjoint subsets of $\Sigma$.

\noindent\textbf{ED text $\tilde{T}$.} 
The text $\tilde{T}$ consists of three parts.
Its middle part encodes all the entries equal to $1$ in matrices $A$, $B$ and $C$, and consists of three string sets ${\cal X}=$$\mathcal{X}_1 \cdot \mathcal{X}_2 \cdot  \mathcal{X}_3$, 
where: 
\begin{enumerate}
\item $\mathcal{X}_1$ contains all strings of the form \(v(i) x a^j\),
for some $i\in[\N]$, $x\in[s]$ and $j\in[\N/s]$ such that
$A[i,(x-1)\cdot(\N/s)+j]=1$;
\item $\mathcal{X}_2$ contains all strings of the form
\(a^{\N/s-j}\) \(x \$\$ y a^{\N/s-k}\), for some $x,y\in[s]$ and $j,k=\in[\N/s]$ such that $B[(x-1)\cdot(\N/s)+j,(y-1)\cdot(\N/s)+k]=1$, i.e., if the corresponding entry of $B$ is 1;
\item  $\mathcal{X}_3$ contains all strings of the form
\(a^{k} y v(i)\), for some \(i\in[\N]\), \(y\in[s]\) and \(k\in[\N/s]\) such that \(C[(y-1)\cdot(\N/s)+k,i]=1\).
\end{enumerate}

\noindent It is easy to see that $|P(i,x,y)| = \cO(\N/s)$. This implies the following:
\begin{enumerate}
\item The length of the pattern is \(m=\cO(\N\cdot s^2\cdot \N/s) = \cO(\N^2\cdot s)\);
\item The total length of $\mathcal{X}$ is \(||{\cal X}||=\cO(\N\cdot s\cdot \N/s \cdot \N/s+s^2\cdot(\N/s)^2\cdot \N/s+\N\cdot s \cdot \N/s \cdot \N/s)=\cO(\N^3/s)\).
\end{enumerate}

\noindent By the above construction, we obtain the following fact.

\begin{fact}
\label{fact:centralPart}
$P(i,x,y)$ matches ${\cal X}$
if and only if the following holds for some $j,k=1,2,\ldots,\N/s$: 
$$A[i,(x-1)\cdot (\N/s)+j]=B[(x-1)\cdot(\N/s)+j,(y-1)\cdot(\N/s)+k]=C[(y-1)\cdot(\N/s)+k,i]=1$$
\end{fact}

Solving the TD problem thus reduces to taking the disjunction of all such conditions.
Let us write down all strings $P(i,x,y)$ in some arbitrary but fixed order to obtain \(P=P_1 P_2 \ldots P_z\) with $z=\N s^2$, where every \( P_t = P(i,x,y)\), for some $i,x,y$. We aim to construct a small number of sets of strings that, when considered as an ED text, match any
prefix $P_1 P_2 \dots P_t$ of the pattern, $1\le t \le z-1$; a similar construction can be carried on to obtain sets of strings that match any suffix \(P_k \dots P_{z-1}P_z\), $2\le k\le z$. These sets will then be added to the left and to the right of $\mathcal{X}$, respectively, to obtain the ED text $\tilde{T}$.

\noindent\textbf{ED Prefix.} 
We construct $\log z$ sets of strings as follows. The first one contains the empty string $\varepsilon$ and $P_1P_2 \dots P_{z/2}$. The second one contains $\varepsilon$, $P_1P_2 \dots P_{z/4}$ and $P_{z/2+1} \dots P_{z/2+z/4}$. The third one contains $\varepsilon$, \( P_1P_2 \dots P_{z/8}\), \(P_{z/4+1} \dots P_{z/4+z/8}\), \(P_{z/2+1}\dots P_{z/2+z/8}\) and $P_{z/2+z/4+1} \dots P_{z/2+z/4+z/8}$. 
\\  Formally,
for every $i=1,2,\dots,\log z$, the $i$-th of such sets is:
\[\tilde{T}^p_i=\varepsilon \cup \{P_{j\frac{z}{2^{i-1}}+1}\dots P_{j\frac{z}{2^{i-1}}+\frac{z}{2^i}} \mid j=0,1,\dots,2^{i-1}-1 \} .\]
\textbf{ED Suffix.}
We similarly construct $\log z$ sets to be appended to $\mathcal{X}$:
\[\tilde{T}^s_{i}=\varepsilon \cup \{P_{z-j\frac{z}{2^{i-1}}-\frac{z}{2^i}+1}\dots P_{z-j\frac{z}{2^{i-1}}} \mid j=0,1,\dots,2^{i-1}-1 \} . \]
The total length of all the ED prefix and ED suffix strings is $\cO(\log z\cdot \N^2\cdot s) = \cO(\N^2\cdot s \cdot \log \N)$. 
The whole ED text $\tilde{T}$ is thus:
$\tilde{T}=\tilde{T}^p_1\cdot\dots \cdot\tilde{T}^p_{\log z}\cdot\mathcal{X}\cdot\tilde{T}^s_{\log z}\cdot\dots\cdot\tilde{T}^s_{1}$.
We next show how a solution of such instance of EDSM corresponds to the solution of TD.

\begin{lemma}\label{lem:match}
The pattern $P$ occurs in the ED text $\tilde{T}$ if and only if there exist $i,j,k$ such that $A[i,j]=B[j,k]=C[k,i]=1$.
\end{lemma}

\begin{proof}
By Fact~\ref{fact:centralPart}, if such $i,j,k$ exist then $P_{t}$ matches ${\cal X}$, for some $t\in\{1,\dots,z\}$. 
Then, by construction of the sets $\tilde{T}^{p}_{i}$ and $\tilde{T}^{s}_{i}$, the prefix $P_{1}\dots P_{t-1}$ matches the ED prefix
(this can be proved by induction), and similarly the suffix $P_{t+1}\dots P_{z}$ matches the ED suffix,
so the whole $P$ matches $\tilde{T}$, and so $P$ occurs therein.
Because of the letters \$ appearing only in the center of $P_{i}$s and strings from ${\cal X}_{2}$,
every $P_{i}$s and a concatenation of $X_{1}\in {\cal X}_{1}$, $X_{2}\in {\cal X}_{2}$, $X_{3}\in {\cal X}_{3}$ having
the same length, and the $P_{i}$s being distinct, there is
an occurrence of the pattern $P$ in $\tilde{T}$ if and only if $X_{1}X_{2}X_{3}=P_{t}$ for some $t$ and
$X_{1}\in {\cal X}_{1}$, $X_{2}\in {\cal X}_{2}$, $X_{3}\in {\cal X}_{3}$. But then, by Fact~\ref{fact:centralPart}
there exists a triangle.
\end{proof}
Note that for the EDSM problem we have $m\!=\!\N^2\cdot s$, $n=\!1\!+\!2\!\log z$ and $N=||{\cal X}||+\cO(\N^2\!\cdot\!s\!\cdot\!\log \N)$. Thus if we had a solution running in $\cO(\log z\cdot m^{1.5-\epsilon}+||{\cal X}||\!+\!\N^2\cdot s\cdot \log \N)\!=\!
\cO(\log\N\cdot(\N^2\cdot s)^{1.5-\epsilon}+\N^3/s)$ time, for some $\epsilon>0$, 
by choosing a sufficiently small $\alpha > 0$ and setting $s=\N^{\alpha}$ we would obtain, for some $\delta > 0 $, an \(\cO(\N^{3-\delta})\)-time algorithm
for TD.
\end{proof}

\section{An $\ctO(nm^{\omega-1}+N)$-time Algorithm for EDSM}
\label{sec:algo}

Our goal is to design a non-combinatorial 
$\ctO(nm^{\omega-1}+N)$-time algorithm for EDSM, which in turn can be achieved  with a non-combinatorial $\ctO(m^{\omega-1}+N)$-time algorithm for the AP problem, that is the bottleneck of EDSM (cf.~\cite{DBLP:conf/cpm/GrossiILPPRRVV17}).

We reduce AP to a logarithmic number of restricted instances of the same problem, based on the length of the strings in ${\cal S}$. We start by giving a lemma that we will use to process na\"ively the strings of length up to a constant $c$, to be determined later, in $\cO(m+N)$ time.

\begin{lemma}\label{lem:naive}
For any integer $t$, all strings in ${\cal S}$ of length at most $t$ can be processed in $\cO(m\log m+mt+N)$ time.
\end{lemma}

\begin{proof}
We first construct the suffix tree $ST$ of $P$ and store, for every node, the first letters on its outgoing edges in a static
dictionary with constant access time. This can be done in $\cO(m\log m)$ time~\cite{ruzic08constructing}.
For every $S\in{\cal S}$, find and mark its corresponding (implicit or explicit) node of $ST$. This takes $\cO(N)$ time overall. For every possible length $t'\!\leq\! t$, scan $P$ with a window of length $t'$ while maintaining
its corresponding node of $ST$. This takes $\cO(m)$ time overall. If the current window $P[i\dd (i+t'-1)]$ corresponds to a marked node of $ST$ and additionally $U[i-1]=1$, we set $V[i+t'-1]=1$.
\end{proof}

We build the rest of the restricted instances of the AP problem by restricting on strings in ${\cal S}_k \subseteq {\cal S}$ of length in $[(19/18)^k,(19/18)^{k+1})$ for each integer $k$ ranging from $\left\lceil\frac{\log c}{\log (19/18)}\right\rceil$ to $\left\lfloor\frac{\log m}{\log (19/18)}\right\rfloor$. These intervals are a partition of the set of all strings in ${\cal S}$ of length up to $m$; longer strings are not addressed in EDSM by solving AP.

For each integer $k$ from $\left\lceil\frac{\log c}{\log (19/18)}\right\rceil$ to $\left\lfloor\frac{\log m}{\log (19/18)}\right\rfloor$, let $\ell$ be an integer such that the length of every string in ${\cal S}_k$
belongs to $[9/8 \cdot \ell,5/4\cdot\ell)$. 
Note that such an integer always exists for an appropriate choice of the integer constant $c$. In fact, it must hold that
\[\frac{9}{8}\cdot\ell~\le~\left(\frac{19}{18}\right)^k<~\left(\frac{19}{18}\right)^{k+1}\le~\frac{5}{4}\cdot\ell \iff \frac{4}{5}\cdot\left(\frac{19}{18}\right)^{k+1}\le~\ell~\le~\frac{8}{9}\cdot\left(\frac{19}{18}\right)^{k}.\]
To ensure that there exists an \emph{integer} $\ell$ satisfying such conditions, it must actually hold that
\[\frac{4}{5}\cdot\left(\frac{19}{18}\right)^{k+1}+1~\le~\frac{8}{9}\cdot\left(\frac{19}{18}\right)^{k}\iff \frac{45}{2}~\le~\left(\frac{19}{18}\right)^{k}.\]
The last equation holds for $k\ge 58$, implying that we will process na\"ively the strings of length up to $c=23$, and each ${\cal S}_k$, for $k$ ranging from $58$ to $\left\lfloor\frac{\log m}{\log (19/18)}\right\rfloor$, will be processed separately as described in the next paragraph.

Denoting by $N_k$ the total size of strings in ${\cal S}_k$, we have that, if 
we solve every such instance of AP in $\cO($$N_k$$+f(m))$ time, then we can solve the original instance of AP in $\cO(N+f(m)\log m)$ time by taking the results disjunction. 
Switching to $\ctO$ notation that disregards polylog factors, it thus suffices to solve each such instance of the AP problem in $\ctO(N+m^{\omega-1})$ time.

We further partition the strings in ${\cal S}_k$ into three types, compute the corresponding bit vector $V$ for each type separately 
and, finally, take the disjunction of the resulting bit vectors $V$ to obtain the answer for each restricted instance.

\medskip
\noindent \textbf{Partitioning ${\cal S}_k$.} 
Keeping in mind that from now on (until Section~\ref{ssec:wrap}) we address the AP problem assuming that ${\cal S}$ only contains strings of length in $[9/8 \cdot \ell,5/4\cdot \ell)$, and thus is in fact ${\cal S}_k$, to lighten the notation we now switch back to denote it simply with ${\cal S}$.
The three types of strings are as follows:
\begin{description}
\item[Type 1:] Strings $S\in {\cal S}$ such that every length-$\ell$ substring of $S$ is not strongly periodic.
\item[Type 2:] Strings $S\in {\cal S}$ containing at least one length-$\ell$ substring that is not strongly periodic and at least one length-$\ell$ substring that is strongly periodic.
\item[Type 3:] Strings $S\in {\cal S}$ such that every length-$\ell$ substring of $S$ is strongly periodic (in Lemma~\ref{lem:type3} we show that in this case $\per(S)\leq \ell/4$).
\end{description}
These three types are evidently a partition of ${\cal S}$. We start with showing that, in fact, strings of type 3 are exactly strings with period at most $\ell / 4$.

\begin{lemma}\label{lem:type3}
Let $S$ be a string. If $\per(S[j\dd j+\ell-1]) \leq \ell/4$ for every $j$ then $\per(S) \leq \ell/4$.
\end{lemma}
\begin{proof}
We first show that, for any string $W$ and letters $a,b$, 
if $\per(aW) \leq |aW|/4$ and $\per(Wb) \leq |Wb|/4$ then $\per(aW)=\per(Wb)$. 
This follows from Lemma~\ref{lem:periodicity}: since $\per(aW)$ and $\per(Wb)$ are both periods of $W$ and $(1+|W|)/4 \leq |W|/2$, then we have that $d=\gcd(\per(aW),\per(Wb))$ is a period of $W$. Assuming by contradiction that $\per(aW)\neq \per(Wb)$, then it must be that either $d<\per(aW)$ or $d<\per(Wb)$; by symmetry it is enough to consider the former possibility, and we claim that then $d$ is a period of $aW$. Indeed, $a=W[\per(aW)-1]$ (observe that $\per(aW)-1\leq |W|$) and $W[i]=W[i+d]$ for any $i=1,2,\ldots,|W|-d$, so by $\per(aW)$ being a multiple of $d$ we obtain that $a=W[\per(aW)-1]=W[d-1]$, which is a contradiction because 
by definition of $\per(aW)$ we have that 
$d<\per(aW)$ cannot be a period of $aW$.

If $\per(S[j\dd j+\ell-1]) \leq \ell/4$ for every $j$ then by the above reasoning the periods of all substrings $S[j\dd j+\ell-1]$ is the same and in fact equal to $p$. But then $S[i]=S[i+p]$ for every $i$, so $\per(S) \leq \ell/4$.
\end{proof}

Before proceeding with the algorithm, we show that, for each string $S\in {\cal S}$,  we can determine its type in $\cO(|S|)$ time.

\begin{restatable}{lemma}{lemdettype}\label{lem:dettype}
Given a string $S$ we can determine its type in $\cO(|S|)$ time.
\end{restatable}
\begin{proof}
It is well-known that $\per(T)$ can be computed in $\cO(|T|)$ time for any string $T$ (cf.~\cite{DBLP:books/daglib/0020103}).
We partition $S$ into blocks $T_{\alpha}=S[\alpha \lfloor \ell/2 \rfloor \dd  (\alpha+1)\lfloor \ell/2 \rfloor-1]$ of size $\lfloor\ell/2\rfloor$,
and compute $\per(T_{\alpha})$
for every $\alpha$ in $\cO(|S|)$ total time. Observe that every substring $S[i\dd i+\ell-1]$ contains at least one whole block inside. 

If $\per(T_{\alpha}) > \ell/4$ then the period of any substring $S[i\dd i+\ell-1]$ that contains $T_{\alpha}$ is also larger than $\ell/4$. Consequently, if $\per(T_{\alpha}) > \ell/4$ for every $\alpha$, then we declare $S$ to be of type 1.

Consider any $\alpha$ such that $p=\per(T_{\alpha}) \leq \ell/4$. If the period $p'$ of a substring
$S'=S[i\dd i+\ell-1]$ that contains $T_{\alpha}$ is at most $\ell/4$, then in fact it must be equal to $p$,
because $p' \geq p$ and so, by Lemma~\ref{lem:periodicity} applied on $T_{\alpha}$, 
$p'$ must be a multiple of $p$ and, by repeatedly applying $S'[j]=S'[j+p']$ and $T_{\alpha}[j]=T_{\alpha}[j+p]$ and using the fact that $T_{\alpha}$
occurs inside $S'$, we conclude that in fact $S'[j]=S'[j+p]$ for any $j$, and thus $p'=p$. 
This allows us to check whether there exists a substring $S'=S[i\dd i+\ell-1]$ that contains $T_{\alpha}$ such that $\per(S')\leq \ell/4$ by computing, in $\cO(\ell)$ time, how far the period $p$ extends to the left and to the right of $T_{\alpha}$ in $T_{\alpha-1}T_{\alpha}T_{\alpha+1}$ (if either $T_{\alpha-1}$ or $T_{\alpha+1}$ do not exist, then we do not extend the period in the corresponding direction). There exists such a substring $S'$ if and only if the length of the extended substring with period $p$ is at least $\ell$. Therefore, for every $\alpha$ we can check in $\cO(\ell)$ time if there exists a length-$\ell$ substring $S'$ containing $T_{\alpha}$ with $\per(S') \leq \ell/4$. By repeating this procedure for every $\alpha$, we can distinguish between $S$ of type 2 and $S$ of type 3 in $\cO(|S|)$ total time.
\end{proof}

Since we have shown how to efficiently partition the strings of {\cal S} into the three types, in what follows we present our solution of the AP problem for each type of strings  separately.

\begin{remark}
\label{rem:numlengthell}
The length of every string in ${\cal S}$ belonging to $[9/8 \cdot \ell,5/4\cdot\ell)$
 implies that every string in ${\cal S}$ contains at most $\ell/4$ length-$\ell$ substrings (and at least $1+\ell/8$ of them).
\end{remark}

\subsection{Type 1 Strings}\label{sec:type1}

In this section we show how to solve a restricted instance of the AP problem where every string $S\in{\cal S}$ is of type 1, that is, each of its length-$\ell$ substrings is not strongly periodic, and furthermore $|S|\in [9/8\cdot \ell,5/4\cdot\ell)$ for some $\ell \leq m$. Observe that all (and hence at most $\ell/4$ by Remark~\ref{rem:numlengthell}) length-$\ell$ substrings of any $S\in{\cal S}$ must be distinct, as otherwise we would be able to find two occurrences of a length-$\ell$ substring at distance at most $\ell/4$ in $S$, making the period of the substring at most $\ell/4$ and contradicting the assumption that $S$ is of type 1.

We start with constructing the suffix tree $ST$ of $P$ (our pattern in the EDSM problem) in $\cO(m\log m)$ time~\cite{DBLP:conf/focs/Weiner73}. Let us
remark that we are spending $\cO(m\log m)$ time and not just $\cO(m)$ so as to avoid any assumptions on the size of the alphabet. For every explicit node $u\in ST$, we construct a perfect hash function mapping the first letter on every edge outgoing from $u$ to the corresponding edge. This takes $\cO(m\log m)$ time~\cite{ruzic08constructing}
and allows us to navigate in $ST$ in constant time per letter. Then, for every $S\in{\cal S}$, we check  in $\cO(|S|)$ time using $ST$ if it occurs in $P$ and, if not, we disregard it from
further consideration. 
Therefore, from now on we assume that all strings $S$, and thus all their length-$\ell$ substrings, occur in $P$.
We will select a set of length-$\ell$ substrings of $P$, called the \emph{anchors}, each represented by one of its occurrences in $P$, such that:
\begin{enumerate}
\item The total number of occurrences of all anchors in $P$ is $\cO(m/\ell\cdot \log m)$.
\item For every $S\in {\cal S}$, at least one of its length-$\ell$ substrings is an anchor. 
\item The total number of occurrences of all anchors in strings $S\in{\cal S}$ is $\cO(|{\cal S}|\cdot \log m)$.
\end{enumerate}
We formalize this using the following auxiliary problem, which is a strengthening of a well-known {\em Hitting Set} problem, which given a collection of $m$ sets over $[n]$, each of size at least $k$, asks to choose a subset of $[n]$ of size
$\cO(n/k\cdot \log m)$ that nontrivially intersects every set.

{\defproblem{\textsc{Node Selection} (NS)}{A bipartite graph $G=(U,V,E)$ with $\deg(u)\in(d,2d]$ for every $u\in U$ and weight $w(v)$ for every $v\in V$, where $W=\sum_{v\in V}w(v)$.}
{A set $V'\subseteq V$ of total weight $\cO(W/d\cdot \log |U|)$ such that $N[u]\cap V' \neq \emptyset$ for every node $u\in U$, and $\sum_{u\in U} |N[u] \cap V'| = \cO(|U|\log |U|)$.}

We reduce the problem of finding anchors to an instance of the NS problem, by building} a bipartite graph $G$
in which the nodes in $U$ correspond to strings $S\in{\cal S}$, the nodes in $V$
correspond to distinct length-$\ell$ substrings of $P$, and there is an edge $(u,v)$ if the length-$\ell$ string corresponding to $v$ occurs in the string $S$ corresponding to $u$.
Using suffix links, we can find the node of the suffix tree corresponding to every length-$\ell$ substring of $S$ in $\cO(|S|)$
total time, so the whole construction takes $\cO(m\log m+\sum_{S\in{\cal S}}|S|)=\cO(m\log m+N)$ time.
The size of $G$ is $\cO(m+N)$, and the degree of every node in $U$ belongs to $(\ell/8,\ell/4]$.
We set the weight of a node $v\in V$ to be its number of occurrences in $P$,
and solve the obtained instance of the NS problem to obtain the set of anchors.
It is not immediately clear that an instance of the NS problem always has a solution. We show that indeed it does, and that it
can be found in linear time.

\begin{lemma}
\label{lem:nodeselectrandom}
A solution to an instance of the NS problem always exists and can be found in linear time in the size of $G$.

\end{lemma}
\begin{proof}
We first show a solution that uses the probabilistic method and leads us to an efficient Las Vegas algorithm; we will then derandomize the solution using the method of conditional expectations.

We independently choose each node of $V$ with probability $p$ to obtain the set $V'$ of selected nodes.
The expected total weight of $V'$ is $\sum_{v\in V} p\cdot w(v)=p\cdot W$, so by Markov's inequality
it exceeds $4p\cdot W$ with probability at most $1/4$.
For every node $u\in U$, the probability that $N[u]$ does not intersect $ V'$ is at most $(1-p)^{d}\leq e^{-pd}$.
Finally, $\mathbb{E}[\sum_{u\in U} |N[u] \cap V'|] \leq |U|\cdot 2pd$, so by Markov's inequality 
$\sum_{u\in U} |N[u] \cap V'|$ exceeds $|U|\cdot 8pd$ with probability at most $1/4$.
We set $p=\ln(4|U|)/d$ (observe that if $p>1$ then we can select all nodes in $V$).
By union bound, the probability that $V'$ is not a valid solution is at most $3/4$, so indeed
a valid solution exists.
Furthermore, this reasoning gives us an efficient Las Vegas algorithm that chooses $V'$ randomly
as described above and then verifies if it constitutes a valid solution. Each iteration takes linear time in the size
of $G$, and the expected number of required iterations is constant.

To derandomize the above procedure we apply the method of conditional expectations. Let $X_{1},X_{2},\ldots$
be the binary random variables corresponding to the nodes of $V$.
Recall that in the above proof we set
$X_{i}=1$ with probability $p$. Now we will choose the values of $X_{1},X_{2},\ldots$ one-by-one.
Define a function $f(X_{1},X_{2},\ldots)$ that bounds the probability that $X_{1},X_{2},\ldots$
corresponds to a valid solution as follows:
\[
f(X_{1},X_{2},\ldots) = \frac{\sum_{v}X_{v}\cdot w(v)}{4W/d\cdot\ln(4|U|)} + \sum_{u\in U}\prod_{v\in N[u]} (1-X_{v}) + \frac{\sum_{u\in U}\sum_{v\in N[u]}X_{v}}{8|U|\ln(4|U|)}.
\]
As explained above, we have $\mathbb{E}[f(X_{1},X_{2},\ldots)]=3/4$. Assume that we have already fixed
the values $X_{1}=x_{1},\ldots,X_{i}=x_{i}$. Then there must be a choice of $X_{i+1}=x_{i+1}$
that does not increase the expected value of $f(X_{1},X_{2},\ldots)$ conditioned on the already chosen values.
We want to compare the following two quantities:
\begin{align*}
\mathbb{E}[f(X_{1},X_{2},\ldots)\,|\,X_{1}=x_{1},\ldots,X_{i}=x_{i},X_{i+1}=0]\\
\mathbb{E}[f(X_{1},X_{2},\ldots)\,|\,X_{1}=x_{1},\ldots,X_{i}=x_{i},X_{i+1}=1]
\end{align*}
and choose $x_{i+1}$ corresponding to the smaller one. Cancelling out the shared terms,
we need to compare the expected values of:
\begin{alignat*}{3}
&~~~~~~~~~~0 &&+~\sum_{u\in N[i+1]}\prod_{v\in N[u]} (1-X_{v})~ &&+~~~~~~~~0~~~~~~~~~~~~\text{and} \\
&\frac{w(i+1)}{4W/d\cdot\ln(4|U|)}~ &&+ ~~~~~~~~~~~~~~~~0 &&+~ \frac{\deg(i+1)}{8|U|\ln(4|U|)} .
\end{alignat*}
The second quantity can be computed in constant time. We claim that (ignoring the issue of numerical precision)
the first quantity can be computed in time $\cO(\deg(i+1))$ after a linear-time preprocessing as follows.
In the preprocessing we compute and store $E[i] =\mathbb{E}[\prod_{j=1}^{i}(1-Y_{j})]$, where 
the $Y_{j}$'s are independent indicator variables with $\Pr[Y_{j}=1]=p$, for every $i=0,1,\ldots,|V|$.
It is straightforward to compute $E[i+1]$ from $E[i]$ in constant time.
Then, during the computation we maintain, for every $u\in U$, the number $c[u]$ of $v\in N[u]$ for which
we still need to choose the value $X_{e}$, and a single bit $b[u]$ denoting whether for some $v\in N[u]\cap\{1,\ldots,i\}$
we already have $x_v=1$. This information can be updated in $\cO(\deg(i+1))$ time after selecting $x_{i+1}$.
Now to compute the first quantity, we iterate over $u\in N[i+1]$ and, if $b[u]=0$ then
we add $E[c[u]]$ to the result.
Finally, we claim that it is enough to implement all calculations with precision $\Theta(\log |V|)$ bits.
This is because such precision allows us to calculate both quantities with relative accuracy $1/(8|V|)$, so
the expected value of $f(X_{1},X_{2},\ldots)$ might increase by a factor of $(1+1/(4|V|))$ in every step,
which is at most $(1+1/(4|V|))^{|V|} \leq e^{1/4}$ overall. This still guarantees that the final value is
at most $3/4\cdot e^{1/4}<1$, so we obtain a valid solution.
\end{proof}

In the rest of this section we explain how to compute the bit vector $V$ from the bit vector $U$, and thus solve the AP problem, after having obtained a set ${\cal A}$ of anchors.
For any $S\! \in \!$ ${\cal S}$, since $S$ contains an occurrence of at least one anchor $H\!\! \in \!\!{\cal A}$, say $S[j\dd \!(j\!+\!|H|\!-\!1)]\!\!=\!\!H$, 
so any occurrence of $S$ in $P$ can be generated by choosing some occurrence of $H$ in $P$, say $P[i\dd (i+|H|-1)]\!=\!H$, and then checking that 
$S[1\dd (j-1)]=P[(i-j+1)\dd (i-1)]$ 
and $S[(j+|H|)\dd |S|]=P[(i+|H|)\dd (i+|S|-j)]$. 
In other words, $S[1\dd (j-1)]$ should be a suffix of $P[1\dd (i-1)]$ and $S[(j+|H|)\dd |S|]$
should be a prefix of $P[(i+|H|)\dd |P|]$. In such case, we say that the occurrence of $S$ in $P$ is generated by $H$. By the properties of ${\cal A}$, any occurrence of $S\in$ ${\cal S}$ is generated by $occ_{S} \geq 1$ occurrences of anchors, where $\sum_{S\in{\cal S}} occ_{S}=\cO(|{\cal S}|\log m)$.
For every $H\!\in \!{\cal A}$ we create a separate data structure $D(H)$ responsible for setting $V[i+|S|\!-\!1]\!=\!1$, when 
$U[i-1]\!=\!1$ and $P[i\dd \!(i\!+\!|S|\!-\!1)]\!=\!S$ is generated by $H$. We now first describe what information is used to initialize each $D(H)$, and how this is later processed to update $V$.

\medskip
\noindent\textbf{Initialization. } $D(H)$ consists of two compact tries $T(H)$ and $T^{r}(H)$. 
For every occurrence of $H$ in $P$, denoted by $P[i\dd (i+|H|-1)]=H$, $T(H)$ should contain a leaf corresponding to $P[(i+|H|)\dd |P|]\$$ and $T^{r}(H)$
should contain a leaf corresponding to $(P[1\dd (i-1)])^{r}\$$, both decorated with position $i$.
Additionally, $D(H)$ stores a list $L(H)$ of pairs of nodes $(u,v)$, where $u\in T^{r}(H)$ and $v\in T(H)$. Each such pair corresponds to an occurrence of $H$ in a string $S\in$ ${\cal S}_{h}$, $S[j\dd (j+|H|-1)]=H$, where $u$ is the node of $T^{r}(H)$ corresponding to $(S[1\dd (j-1)])^{r}\$$ and $v$ is the node of $T(H)$ corresponding to $S[(j+|H|+1)\dd |S|]\$$. We claim that $D(H)$, for all $H$, can be constructed in $\cO(m\log m+N)$ total time. 

We first construct the suffix tree $ST$ of $P\$$ and the suffix tree $ST^{r}$ of $P^{r}\$$ (again in $\cO(m\log{m})$ time not to make assumptions on the alphabet).
We augment both trees with data for answering both {\em weighted ancestor} (WA) and
{\em lowest common ancestor} (LCA) queries,  that are defined as follows.
For a rooted tree $T$ on $n$ nodes with an integer weight $\mathcal{D}(v)$ assigned to every node $u$, such that the weight of the root is zero and
$\mathcal{D}(u) < \mathcal{D}(v)$ if $u$ is the parent of $v$, we say that a node $v$ is a weighted ancestor of a node $v$ at depth $\ell$,
denoted by $\text{WA}_T(u,\ell)$, if $v$ is the highest ancestor of $u$ with weight at least $\ell$.
Such queries can be answered in 
$\cO(\log n)$ time after an $\cO(n)$ preprocessing~\cite{DBLP:conf/cpm/FarachM96}.
For a rooted tree $T$, $\text{LCA}_T(u,v)$ is the lowest node that is an ancestor of both $u$ and $v$.
Such queries can be answered in $\cO(1)$ time after an $\cO(n)$ preprocessing~\cite{DBLP:conf/latin/BenderF00}.
Recall that every anchor $H$ is represented by one of its occurrences in $P$. 
Using WA queries, we can access in $\cO(\log m)$ time the nodes corresponding to $H$ and $H^{r}$, respectively, and extract a lexicographically sorted list of suffixes following an occurrence of $H$ in $P\$$ and a lexicographically sorted list of reversed prefixes preceding an occurrence of $H$ in $P^{r}\$$ in time proportional to the number of such occurrences.
Then, by iterating over the lexicographically sorted list of suffixes and using LCA queries on $ST$ we can build $T(H)$ in time proportional to the length of the list, and similarly we can build $T^{r}(H)$. To construct $L(H)$ we start by computing, for every $S\in$ ${\cal S}$ and $j=1,\ldots,|S|$, the node of $ST^{r}$ corresponding to $(S[1\dd j])^{r}$ 
and the node of $ST$ corresponding to $S([(j+1)\dd |S|]$ (the nodes might possibly be implicit). 
This takes only $\cO(|S|)$ time, by using suffix links. 
We also find, for every length-$\ell$ substring $S[j\dd (j+\ell-1)]$ of $S$, an anchor $H\in{\cal A}$ such that $S[j\dd (j+\ell-1)]=H$, if any exists. This can be done by finding the nodes (implicit or explicit) of $ST$ that correspond to the anchors, and then scanning over all length-$\ell$ substrings while maintaining the node of $ST$ corresponding to the current substring using suffix links in $\cO(|S|)$ total time. After having determined that $S[j\dd (j+\ell-1)]=H$ we add $(u,v)$ to $L(H)$, where $u$ and $v$ are the previously found nodes of $ST^{r}$ and $ST$ corresponding to $(S[1\dd (j-1)])^{r}$ and $S[(j+\ell)\dd |S|]$, respectively. By construction, we have the following property, also illustrated in Figure~\ref{fig:occurrence}.

\begin{fact}
\label{fac:occurrence}
A string $S\in$ ${\cal S}$ starts at position $i\!-\!j\!+\!1$ in $P$ if and only if, for some anchor $H\in {\cal A}$, $L(H)$ contains a pair $(u,v)$ corresponding to $S[j\dd (j\!+\!|H|\!-\!1)]\!=\!H$,
such that the subtree of $T^{r}(H)$ rooted at $u$ and that of $T(H)$ rooted at $v$ contain a leaf decorated with $i$.
\end{fact}
\begin{figure}[!t]
\begin{center}
\includegraphics[width=0.8\textwidth]{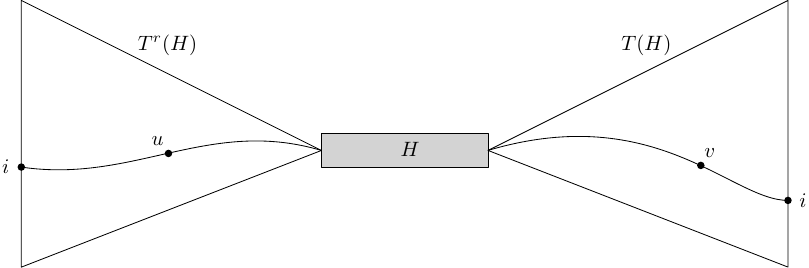}
\end{center}
\caption{An occurrence of $S$ starting at position $i$ in $P$ is generated by $H$: $(u,v)$ corresponds to $S[j\dd (j+|H|-1)]=H$ and $i$ appears in the subtree of $T^r(H)$ rooted at $u$, as well as in the subtree  of $T(H)$ rooted at $v$.}
\label{fig:occurrence}
\end{figure}

Note that the overall size of all lists $L(H)$, when summed up over all $H\in{\cal A}$, is $\sum_{S\in{\cal S}}occ_{S}=\cO(|{\cal S}|\log m)$,
and since each $S$ is of length at least $\ell$ this is $\cO(N/\ell \cdot \log m)$.

\medskip
\noindent\textbf{Processing.} The goal of processing $D(H)$ is to efficiently process all occurrences generated by $H$. As a preliminary step, we decompose $T^{r}(H)$ and $T(H)$ into heavy paths. Then, for every pair of leaves $u\in T^{r}(H)$ and $v\in T(H)$ decorated by the same $i$, 
we consider all heavy paths above $u$ and $v$.
Let $p=u_{1}-u_{2}-\ldots$ be a heavy path above $u$ in $T^{r}(H)$ and $q=v_{1}-v_{2}-\ldots$ be a heavy path
above $v$ in $T(H)$, where $u_{1}$ is the head of $p$ and $v_{1}$ is the head of $q$, respectively.
Further, choose the largest $x$ such that $u$ is in the subtree rooted at $u_{x}$, and the largest $y$ such that $v$ is in the subtree rooted at $v_{y}$ (this is well-defined by the choice of $p$ and $q$, as $u$ is in the subtree rooted at $u_{1}$ and $v$ is in the subtree rooted at $v_{1}$). We add $(i,|\nodestring(u_{x})|,|\nodestring(v_{y})|)$ to an auxiliary list associated with the pair of heavy paths $(p,q)$. In the rest of the processing we work with each such lists separately. Notice that the overall size of all auxiliary lists, when summed up over all $H\in{\cal A}$, is $\cO(m/\ell\cdot \log^{3}m)$, because there are at most $\log^{2} m$ pairs of heavy paths above $u$ and $v$
decorated by the same $i$, and the total number of leaves in all trees $T^{r}(H)$ and $T(H)$ is bounded by the total number of occurrences of all anchors in $P$, which is $\cO(m/\ell\cdot \log m)$. By Fact~\ref{fac:occurrence}, there is an occurrence of a string ${\cal S}$ generated by $H$ and starting at position $i-j+1$ in $P$ if and only if $L(H)$ contains a pair $(u,v)$ corresponding to $S[j\dd (j+|H|-1)]=H$ such that, denoting by $p$
the heavy path containing $u$ in $T^{r}(H)$ and by $q$ the heavy path containing $v$ in $T(H)$,
the auxiliary list associated with $(p,q)$ contains a triple $(i,x,y)$ such that $x\geq |\nodestring(u)|$ and $y\geq |\nodestring(v)|$. This is illustrated in Figure~\ref{fig:occurrence_auxiliary}. Henceforth, we focus onthe problem of processing a single auxiliary list associated with $(p,q)$, together with a list of pairs $(u,v)$, such that $u$ belongs to $p$ and $v$ belongs to $q$.

\begin{figure}[!t]
\begin{center}
\includegraphics[width=0.8\textwidth]{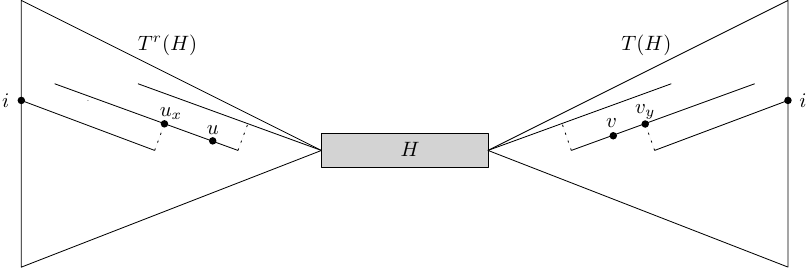}
\end{center}
\caption{An occurrence of $S$ starting at position $i$ in $P$ 
corresponds to a triple $(i,\nodestring(u_{x}),\nodestring(v_{y}))$ on some auxiliary list.}
\label{fig:occurrence_auxiliary}
\end{figure}

An auxiliary list can be interpreted geometrically as follows: for every $(i,x,y)$ we create a red point $(x,y)$, and for every $(u,v)$ we create a blue point $(|\nodestring(u)|,|\nodestring(v)|)$. Then, each occurrence of $S\in{\cal S}$ generated by $H$ corresponds to a pair of points $(p_1,p_2)$ such that $p_1$ is red, $p_2$ is blue, and
$p_1$ dominates $p_2$. We further reduce this to a collection of simpler instances in which all red
points already dominate all blue points. This can be done with a divide-and-conquer procedure
which is essentially equivalent to constructing a 2D range tree~\cite{Bentley:1975:MBS:361002.361007}.
The total number of points in all obtained instances increases by a factor of $\cO(\log^{2}m)$,
making the total number of red points in all instances $\cO(m/\ell\cdot \log^{5} m)$,
while the total number of blue points is $\cO(N/\ell \cdot \log^{3}m)$.
There is an occurrence of a string $S\in$ ${\cal S}$ generated by $H$ and
starting at position $i-j+1$ in $P$ if and only if some simpler instance contains
a red point created for some $(i,x,y)$ and a blue point created for some $(u,v)$ corresponding
to $S[j\dd (j+|H|-1)]=H$. In the following we focus on processing a single simpler instance.

To process a simpler instance we need to check if $U[i-j]=1$, for a red point created for some $(i,x,y)$
and a blue point created for some $(u,v)$ corresponding to $S[j\dd (j+|H|-1)]=H$,
and if so set $V[i-j+|S|]=1$. This has a natural interpretation as an instance of
BMM: we create a 
$\lceil 5/4\cdot \ell\rceil \times \lceil 5/4\cdot \ell \rceil$ matrix $M$ such that
$M[|S|-j,\lceil 5/4\cdot \ell \rceil +1 - j]=1$ if and only if there is a blue point created for some $(u,v)$ corresponding to $S[j\dd (j+|H|-1)]=H$; then for every red point created for some $(i,x,y)$ we construct a bit vector
$U_{i} = U[(i-\lceil 5/4\cdot \ell\rceil)\dd (i-1)]$ (if $i<\lceil5/4\cdot \ell\rceil$, 
we pad $U_{i}$ with $0$s to make its length always equal to $\lceil 5/4\cdot \ell\rceil$);
calculate $V_{i}=M \times U_{i}$; and finally set $V[i+j]=1$ whenever $V_{i}[j]=1$ (and $i+j\leq m$).

\begin{restatable}{lemma}{lemmatrixdef}\label{lem:matrixdef}
$V_{i}[k]=1$ if and only if there is a blue point created for some $(u,v)$ corresponding to $S[j\dd (j+|H|-1)]=H$ such that $U[i-j]=1$ and $k=|S|-j$.
\end{restatable}
\begin{proof}
By definition of $V_{i}=M \times U_{i}$, we have that $V_{i}[k]=1$ if and only if $M[k,t]=1$ for some $t$ such that $U_{i}[t]=1$.
By definition of $U_{i}$, we have that $U_{i}[t]=1$ if and only if
$U[i-\lceil 5/4\cdot \ell \rceil+t-1]=1$, and hence the previous condition
can be rewritten as $M[k,t]=1$ and $U[i-\lceil5/4\cdot \ell\rceil+t-1]=1$, or
equivalently, by substituting $j=\lceil5/4\cdot\ell\rceil+1-t$, $M[k,\lceil5/4\cdot \ell\rceil+1-j]=1$ and $U[i-j]=1$.
By definition of $M$, we have that $M[k,\lceil5/4\cdot \ell\rceil+1-j]=1$ if and only if there is a blue point created for some $(u,v)$ corresponding to $S[j\dd (j+|H|-1)]=H$ with $k=|S|-j$, which proves the lemma.
\end{proof}

The total length of all vectors $U_{i}$ and $V_{i}$ is $\cO(m\log^{5}m)$, so we can afford to extract the appropriate
fragment of $U$ and then update the corresponding fragment of $V$. The bottleneck is computing the matrix-vector product $V_{i}=M \times U_{i}$. Since the total number of $1$s in all matrices $M$ is bounded by the total number of blue points, a na\"ive method would take $\cO(N/ \ell \cdot \log^{3}m)$ time; we
 overcome this by processing together all multiplications concerning the same matrix $M$, thus amortizing the costs. Let $U_{i_{1}},U_{i_{2}},\ldots,U_{i_{s}}$ be all bit vectors
that need to be multiplied with $M$, and let $z$ a parameter to be determined later. We distinguish between two cases:
$(i)$ if \textbf{$s < z$}, then  we compute the products na\"ively by iterating over all $1$s in $M$, and the total computation time, when summed up over all such matrices $M$, is $\cO(N/\ell \cdot \log^{3}m \cdot z)$; 
$(ii)$ if \textbf{$s\geq z$}, then  we partition the bit vectors into $\lceil s/z\rceil \leq s/z+1$ groups of $z$ (padding the last group with bit vectors containing all $0$s)  and, for every group, we create a single matrix whose columns contain 
all the bit vectors 
belonging to the group. Thus, we reduce the problem of computing all matrix-vector products $M \times U_{i}$ to that of computing $\cO(s/z)$ matrix-matrix products of the form $M \times M'$, where $M'$ is an $\lceil5/4\cdot \ell\rceil \times z$ matrix. Even if $M'$ is not necessarily a square matrix, we can still apply the fast matrix multiplication algorithm to compute $M\times M'$ using the standard trick of decomposing the matrices into square blocks.

\begin{lemma}\label{lem:rectangle}
If two $\N \times \N$ matrices can be multiplied in $\cO(\N^{\omega})$ time, then, for any $\N \geq \N'$, an $\N \times \N$ and an
$\N \times \N'$ matrix can be multiplied in $\cO((\N/\N')^{2} \N'^{\omega})$ time.
\end{lemma}
\begin{proof}
We partition both matrices into blocks of size $\N' \times \N'$. There are $(\N / \N')^{2}$ such blocks in the first matrix and $\N / \N'$ in the second matrix. Then, to compute the product we multiply each block from the first matrix by the appropriate block in the second matrix in $\cO(\N'^{\omega})$ time, resulting in the claimed complexity.
\end{proof}

By applying Lemma~\ref{lem:rectangle}, we can compute $M\times M'$ in $\cO(\ell^{2} z^{\omega-2})$ time
(as long as we later verify that $5/4\cdot \ell \geq z$), so
all products $M\times U_{i}$ can be computed in $\cO(\ell^{2}z^{\omega-2}\cdot(s/z+1))$ time. Note that this case can occur only $\cO(m/(\ell\cdot z)\cdot \log^{5}m)$ times, because all values of $s$ sum up to
$\cO(m/\ell\cdot \log^{5}m)$. This makes the total computation time, when summed up over all such matrices $M$, $\cO(\ell^{2}z^{\omega-2} \cdot m/(\ell\cdot z)\cdot \log^{5}m)=\cO(\ell z^{\omega-3}\cdot m\log^{5}m)$.
We can now prove our final result for strings of type 1.

\begin{theorem}\label{thm:type1}

An instance of the AP problem where all strings are of type 1 can be solved in $\ctO(m^{\omega-1}+N)$ time.
\end{theorem}
\begin{proof}
The total time complexity is first $\cO(m+N)$ to construct the graph $G$,
then $\cO(m\log m+N)$ to solve its corresponding instances of the \textsc{NodeSelection} problem
and obtain the set of anchors $H$.
The time to initialize all structures $D(H)$ is $\cO(m\log m+N)$. For every $D(H)$, we obtain in $\cO(m/\ell \cdot \log^{5}m+N/\ell\cdot \log^{3}m)$ time a number of simpler instances, and then construct the corresponding Boolean matrices $M$ and bit vectors $U_{i}$ in additional $\cO(m\log^{5}m)$ time. Note that some $M$ might be sparse, so we need to represent them as a list of $1$s. Then, summing up over all matrices $M$ and both cases, we spend $\cO(N / \ell \cdot \log^{3}m \cdot z +\ell z^{\omega-3}\cdot m\log^{5}m )$ time. We would like to assume that $\ell \geq \log^{3}m$ so that we can set $z=\ell / \log^{3}m$. This is indeed possible, because for any $t$ we can switch to a more na\"ive approach to process all strings of length at most $t$ in $\cO(mt^{2}+N)$ time as described in~\ref{lem:naive}.
After applying it with $t=\log^{3}m$ in $\cO(m\log^{6}m+N)$ time, we can set $z=\ell / \log^{3}m$ (so that indeed $5/4\cdot \ell \geq z$ as required in case $s\geq z$) and the overall time complexity for all matrices $M$ and both cases becomes $\cO(N+\ell^{\omega-2}\cdot m\log^{5+3(3-\omega)}m)$. 
Summing up over all values of $\ell$ and taking the initialization into account we obtain $\cO(m\log^{7}m+m^{\omega-1}\log^{5+3(3-\omega)}m+N)=\ctO(m^{\omega-1}+N)$ total time.
\end{proof}
\subsection{Type 2 Strings}\label{app:type2}

In this section we show how to solve a restricted instance of the AP problem where every string $S\in{\cal S}$ is of type 2, that is, $S$ contains a length-$\ell$ substring that is not strongly periodic as well as a length-$\ell$ substring that is strongly periodic, and furthermore $|S|\in [9/8\cdot \ell,5/4\cdot\ell)$ for some $\ell \leq m$.

Similarly as in Section~\ref{sec:type1}, we select a set of anchors. In this case, instead of the \textsc{NodeSelection} problem we need to exploit periodicity. We call a string $T$ \emph{$\ell$-periodic} if $|T|\geq \ell$ and $\per(T) \leq \ell/4$. We consider all maximal $\ell$-periodic substrings of $S$, that is, $\ell$-periodic substrings $S[i\dd j]$ such that either $i=1$ or $\per(S[(i-1)\dd j]) > \ell/4$, and $j=|S|$ or $\per(S[i\dd (j+1)]) > \ell/4$.
We know that $S$ contains at least one such substring (because there exists a length-$\ell$ substring
that is strongly periodic), and that the whole $S$ is not such a substring (because otherwise $S$ would be of type 3). Further, two maximal $\ell$-periodic substrings cannot overlap too much, as formalized in the following lemma.

\begin{lemma}
\label{lem:periodicoverlap}
Any two distinct maximal $\ell$-periodic substrings of the same string $S$ overlap by less than $\ell/2$ letters.
\end{lemma}
\begin{proof}
Assume (by contradiction) the opposite; then we have two distinct $\ell$-periodic substrings $S[i\dd j]$ and $S[i'\dd j']$ such that 
$i < i' \leq j < j'$ and $j-i'+1 \geq \ell/2$.
Then, both $p=\per(S[i\dd j])$ and $p'=\per(S[i'\dd j'])$ are periods of $S[i'\dd j]$, and hence by Lemma~\ref{lem:periodicity} we have that 
$\gcd(p,p')$ is a period of $S[i'\dd j]$. If $p\neq p'$ then, because $S[i'\dd j]$ contains an occurrence of both $S[i\dd (i+p-1)]$
and $S[i'\dd (i'+p'-1)]$, we obtain that one of these two substrings is a power of a shorter string, thus contradicting
the definition of $p$ or $p'$.
So $p=p'$, but then $p\leq \ell/4$ is actually a period of the whole $S[i\dd j']$, meaning that $S[i\dd j]$ and $S[i'\dd j']$ are not
maximal, a contradiction.
\end{proof}

By Lemma~\ref{lem:periodicoverlap}, every $S\in{\cal S}$ contains exactly one maximal $\ell$-periodic substring, and 
by the same argument $P$ contains $\cO(m/\ell)$ such substrings. 
The set of anchors will be generated by considering the unique maximal $\ell$-periodic substring of every $S\in{\cal S}$,
so we first need to show how to efficiently generate such substrings.

\begin{lemma}
\label{lem:ellperiodic}
Given a string $S$ of length at most $5/4\cdot \ell$, we can generate its (unique) maximal $\ell$-periodic substring in $\cO(|S|)$ time.
\end{lemma}
\begin{proof}
We start with observing that any length-$\ell$ substring of $S$ must contain
$S[(\lfloor \ell/2 \rfloor +1) \dd \ell]$ inside.
Consequently, we can proceed similarly as in the proof 
 of Lemma~\ref{lem:dettype}.
We compute $p=\per(S[(\lfloor \ell/2 \rfloor+1)\dd \ell])$ in $\cO(|S|)$ time.
If $p > \ell/4$ then $S$ does not contain any $\ell$-periodic substrings.
Otherwise, we compute in $\cO(|S|)$ time how far the period $p$ extends to the left and to the right;
that is, we compute the smallest $i \leq \lfloor\ell/2\rfloor+1$ such that $S[k]=S[k+p]$ for every
$k=i,i+1,\ldots,\lfloor \ell/2 \rfloor$ and the largest $j \geq \ell$ such that $S[k]=S[k-p]$ for every
$k=\ell+1,\ell+2,\ldots,j$. If $j-i+1 \geq \ell$ then $S[i\dd j]$ is a maximal $\ell$-periodic substring of $S$,
and, as shown earlier by Lemma~\ref{lem:periodicoverlap}, $S$ cannot contain any other
maximal $\ell$-periodic substrings. We return $S[i\dd j]$ as the (unique) maximal $\ell$-periodic substring of $S$.
\end{proof}

For every $S\in{\cal S}$, we apply Lemma~\ref{lem:ellperiodic} on $S$ to find its (unique) maximal $\ell$-periodic substring $S[i\dd j]$ in $\cO(|S|)$ time. If $i>1$ then we designate $S[(i-1)\dd (i-1+\ell)]$ as an anchor, and similarly if $j<|S|$ we designate $S[(j+1-\ell)\dd (j+1)]$ as an anchor.
Observe that because $S$ is of type 2 (and not of type 3) either $i>1$ or $j<|S|$, so for every $S\in{\cal S}$ we designate at least one
if its length-$(\ell+1)$ substrings as an anchor.
As in Section~\ref{sec:type1}, we represent each anchor by one of its occurrences in $P$, and so need to find its corresponding node
in the suffix tree of $P$ (if any). This can be done in $\cO(|S|)$ time, so $\cO(N)$ overall.
During this process we might designate the same string as an anchor multiple times, but we can easily remove the possible
duplicates to obtain the set ${\cal A}$ of anchors in the end.
Then, we generate the occurrences of all anchors in $P$ by accessing their corresponding nodes in the suffix tree of $P$
and iterating over all leaves in their subtrees.
We claim that the total number of all these occurrences is only $\cO(m/\ell)$. This follows from the following characterization.

\begin{lemma}
\label{lem:anchorsoccurrences}
If $P[x\dd (x+\ell)]$ is an occurrence of an anchor then either $P[(x+1)\dd y]$ is a maximal $\ell$-periodic substring of $P$, for some $y \geq x+\ell$,
or $P[x'\dd (x+\ell-1)]$ is a maximal $\ell$-periodic substring of $P$, for some $x' \leq x$.
\end{lemma}
\begin{proof}
By symmetry, it is enough to consider an anchor $H$ created because of a maximal $\ell$-periodic substring $S[i\dd j]$ such that $i>1$,
when we add $S[(i-1)\dd (i-1+\ell)]$ to ${\cal A}$. Thus, $\per(H[2\dd |H|])\leq \ell/4$ and if $P[x\dd (x+\ell)]=H$ then
$\per(P[(x+1)\dd (x+\ell)]) \leq \ell/4$, making $P[(x+1)\dd (x+\ell)]$ a substring of some maximal $\ell$-periodic substring
of $P[(x'+1)\dd y]$, where $x' \leq x$ and $y\geq x+\ell$.
If $x' < x$ then $\per(H)\leq \ell/4$. But then $H=S[(i-1)\dd (i-1+\ell)]$ can be extended to some maximal $\ell$-periodic substring
$S[i'\dd j']$ such that $i'\leq i-1$ and $j' \geq i-1+\ell$. The overlap between $S[i\dd j]$ and $S[i'\dd j']$ is at least $\ell$,
so by Lemma~\ref{lem:periodicoverlap} $i=i'$ and $j=j'$, which is a contradiction. Consequently, $x' = x$ and we obtain the lemma.
\end{proof}

By Lemma~\ref{lem:anchorsoccurrences}, the number of occurrences of all anchors in $P$ is at most two per each
maximal $\ell$-periodic substrings, so $\cO(m/\ell)$ in total. We thus obtain a set of length-$(\ell+1)$
anchors with the following properties:
\begin{enumerate}
\item The total number of occurrences of all anchors in $P$ is $\cO(m/\ell)$.
\item For every $S\in{\cal S}$, at least one of its length-$(\ell+1)$ substrings is an anchor.
\item For every $S\in{\cal S}$, at most two of its length-$(\ell+1)$ substrings are anchors.
\end{enumerate}
These properties are even stronger than what we had used in Section~\ref{sec:type1} (except that now we are working with length-$(\ell+1)$
substrings, which is irrelevant), we can now prove our final result also for strings of type 2.

\begin{restatable}{theorem}{thmtypetwo}\label{thm:type2}
An instance of the AP problem where all strings are of type 2 can be solved in $\ctO(m^{\omega-1}+N)$ time.
\end{restatable}

\subsection{Type 3 Strings}\label{app:type3}

In this section we show how to solve a restricted instance of the AP problem where every string $S\in{\cal S}$ is of type 3, that is, 
$\per(S)\leq \ell/4$. Furthermore $|S|\in [9/8\cdot \ell,5/4\cdot\ell)$ for some $\ell \leq m$.
Recall that strings $S\in {\cal S}$ are such that every length-$\ell$ substring of $S$ is strongly periodic and, by Lemma~\ref{lem:type3}, in this case, $\per(S)\leq \ell/4$.
An occurrence of such $S$ in $P$ must be contained in a maximal $\ell$-periodic substring of $P$.  Recall that a string $T$ is called $\ell$-periodic if $|T|\geq \ell$ and $\per(T) \leq \ell/4$.
For an $\ell$-periodic string $T$, let its \emph{root}, denoted by $\wordroot(T)$, be the lexicographically smallest cyclic shift of $T[1\dd \per(T)]$.
Because $\per(T)\leq \ell/4$ and $|T|\geq\ell$ by definition, there are at least four repetitions of the period in $T$, so we can write $T=R[i\dd |R|] R^{\alpha} R[1\dd j]$, where $R=\wordroot(T)$, for some $i,j\in [1,|R|]$ and $\alpha\geq 2$. It is well known that $\wordroot(T)$ can be computed in $\cO(|T|)$ time~\cite{DBLP:journals/jal/Duval83}.

\begin{example}
Let $T=\textnormal{\texttt{babababab}}$ and $\ell=8$. We have $|T|=9\geq \ell=8$ and $\per(T) = 2 \leq \ell/4 = 2$, so $T$ is $\ell$-periodic. We have $\wordroot(T)=R=\textnormal{\texttt{ab}}$, and $T$ can be written as $T=\textnormal{\texttt{b}}\cdot(\textnormal{\texttt{ab}})^3\cdot\textnormal{\texttt{ab}}$, for $i=2$ and $j=2$.
\end{example}

We will now make a partition of type 3 strings based on their root.
We start with extracting all maximal $\ell$-periodic substrings of $P$ using Lemma~\ref{lem:ellperiodic} and
compute the root of every such substring. This can be done  in $\cO(m)$ total time because two maximal $\ell$-periodic substrings cannot overlap by more than
$\ell/2$ letters, and hence their total length is at most $3/2\cdot \ell$. We also extract the root of every $S\in{\cal S}$ in $\cO(N)$ total time. We then partition maximal $\ell$-periodic substrings of $P$ and strings $S\in{\cal S}$ into groups that have the same root. In the remaining part we describe how to process each such group corresponding to root $R$ in which all maximal $\ell$-periodic substrings of $P$ have total length $m'$, and the strings $S\in{\cal S}$ have total length $N'$.

Recall that the bit vector $U$ stores the active prefixes input to the AP problem, and the bit vector $V$ encodes the new active prefixes we aim to compute. For every maximal $\ell$-periodic substring of $P$ with root $R$ we extract the corresponding fragment of the bit vector $U$ and need to update the corresponding fragment of the bit vector $V$. To make the description less cluttered, we assume that each such substring of $P$ is a power of $R$, that is, $R^{\alpha}$ for some $\alpha \geq 4$. This can be assumed without loss of generality as it can be ensured by appropriately padding the extracted fragment of $U$ and then truncating the results, while increasing the total length of all considered substrings of $P$ by at most half of their length. In the description below, for simplicity of presentation, $U$ and $V$ denote these padded fragments of the original $U$ and $V$. When computing $V$ from $U$ we use two different methods for processing the elements $S=R[i\dd |R|]R^{\beta}R[1\dd j]$ of $\mathcal{S}$ depending on their length: either $\beta > \alpha / |R|$ (large $\beta$) or $\beta \le \alpha / |R|$ (small $\beta$).

\medskip

\noindent\textbf{Large $\beta$.} We proceed in phases corresponding to $\beta=\alpha / |R|+1,\ldots,\alpha$. In each single phase,
we consider all strings $S\in{\cal S}$ with $S=R[i\dd |R|]R^{\beta}R[1\dd j]$, for some $i$ and $j$. Let $C(\beta)$ be the set of the corresponding pairs $(i,j)$, and observe that $\sum_{\beta}|C(\beta)| \cdot |R^{\beta}| \leq N'$.
This is because the length of $R^{\beta}$ is not greater than that of $S=R[i\dd |R|]R^{\beta}R[1\dd j]$, there are $|C(\beta)|$ distinct strings of the latter form, and the total length of all $S\in\cal{S}$ is $N'$. 
The total number of occurrences of a string $S=R[i\dd |R|]R^{\beta}R[1\dd j]$ in $R^{\alpha}$ is bounded by $\cO(\alpha)$, and all such occurrences can be generated in time proportional to their number. Thus, for every $(i,j)\in C(\beta)$,
we can generate all occurrences of the corresponding string and appropriately update $V$ in $\cO(\alpha\cdot |C(\beta)|)$ total time.
\medskip

\noindent\textbf{Small $\beta$.}
We start by giving a technical lemma on the complexity of multiplying two $r\times r$ matrices whose cells are polynomials of degree up to $d$.

\begin{lemma}\label{lem:polymatrix}
If two $r\times r$ matrices over $\mathbb{Z}$ can be multiplied in $\cO(r^{\omega})$ time, then two $r\times r$ matrices over $\mathbb{Z}[X]$ with degrees up to $d$ can be multiplied in $\tilde{\cO}(dr^{\omega})$ time. 
\end{lemma}

\begin{proof}
Let $A$ and $B$ be two $r \times r$ matrices over $\mathbb{Z}[X]$ with degrees up to $d$.
We reduce the product $A\cdot B=C$ to $2d$ products of $r \times r$ matrices over $\mathbb{Z}$ as follows. 
We evaluate the polynomials of each matrix in the complex $(2d)$th roots of unity: let $A_i$ and $B_i$ be the matrices obtained by evaluating the polynomials of $A$ and $B$ in the $i$-th such root, respectively.
We then perform the $2d$ products $A_1\cdot B_1,\ldots,A_{2d}\cdot B_{2d}$ to obtain matrices $C_1,\ldots,C_{2d}$: the $2d$ values $C_1[i,j],\ldots,C_{2d}[i,j]$ are finally interpolated to obtain the coefficient representation of $C[i,j]$, for each $i,j=1,\ldots,r$, in $\cO(d\log d)$ time for each polynomial~\cite{FFT}.
Since we perform $2d$ products of matrices in $\mathbb{Z}^{r\times r}$, and we evaluate and interpolate $r^2$ polynomials of degree up to $2d$, the overall time complexity is $2d\cO(r^{\omega})+r^2\cO(d\log d)=
\tilde{\cO}(dr^{\omega})$.
\end{proof}

Unlike in the large $\beta$ case, we process $\beta=2,\ldots,\alpha/|R|$ simultaneously as follows. For each $\beta$ we construct an $|R|\times |R|$ matrix $M_{\beta}$, with $M_{\beta}[i,j]=1$ if and only if $(i,j)\in C(\beta)$ (and $M_{\beta}[i,j]=0$ otherwise), and collect them in a single 3D matrix $M\in\{0,1\}^{|R|\times|R|\times(\alpha/|R|-1)}$ with the third dimension corresponding to the value of $\beta$. We then create another
$\alpha \times |R|$ matrix, denoted by $M'$, by setting $M'[\gamma,i]=1$ if and only if $U[\gamma\cdot |R|+i-1]=1$ (and $M_{\beta}[i,j]=0$ otherwise).
Observe that $M'$ can be interpreted as a vector of length $|R|$ over $\mathbb{Z}[X]$ with degrees up to $\alpha$, and $M$ as an $|R|\times|R|$ matrix over $\mathbb{Z}[X]$ with degrees up to $\alpha/|R|$: in this way, $x^{\gamma}$ appears with non-zero coefficient in the polynomial at $M'[i]$ if and only if $U[\gamma\cdot |R|+i-1]=1$, and $x^{\beta}$ appears with non-zero coefficient in the polynomial at $M[i,j]$ if and only if $(i,j)\in C(\beta)$.
$M$ can be constructed in total $\cO(N')$ time by first iterating over all $S\in{\cal S}$ and adding $x^{\beta}$ to the polynomial at $M[i,j]$, where $S=R[i\dd |R|]R^{\beta}R[1\dd j]$, and then extracting a prefix of each polynomial consisting of monomials of degree less than $\alpha/|R|$.

The product $M'\cdot M = M''$ allows us to recover the updates to $V$ by observing that $V[(q+1)\cdot |R|+j]=1$ if and only if $x^{q}$ appears with non-zero coefficient in the polynomial at $M''[j]$. We aim at reducing this product to a matrix-matrix product over $\mathbb{Z}[X]$ with degrees up to $\alpha/|R|$, so as to compute it efficiently by applying Lemma~\ref{lem:polymatrix}.

The idea now is to decompose the columns of $M'$ into $|R|$ chunks of size $\alpha/|R|$ in order to transform it into another 3D matrix. To this end, we transform $M'$ into an $|R|\times|R|\times (\alpha/|R|)$ matrix $A$ by setting $$A[k,i,\gamma]=1\Leftrightarrow M'[(k-1)\alpha/|R|+\gamma,i]=1\Leftrightarrow U[(k-1)\alpha/|R|+\gamma+i-1]=1.$$
By interpreting $A$ as an $|R|\times|R|$ matrix over $\mathbb{Z}[X]$ with degrees up to $\alpha/|R|$, and interpreting $M'$ as a vector of length $|R|$ over $\mathbb{Z}[X]$ with degrees up to $\alpha$, we have that the first row of $A$ consists of the coefficients of 
$x^{1},\ldots, x^{\alpha/|R|}$ of each of the $|R|$ polynomials of $M'$, the second row consists of the coefficients of $x^{\alpha/|R|+1},\ldots, x^{2\alpha/|R|}$ of each of the $|R|$ polynomials of $M'$, and so on. In general, $A[k,i]$ consists of the coefficients of $x^{(k-1)\alpha/|R|+1},\ldots, x^{k\alpha/|R|}$ of polynomial $M'[i]$.

The product $A\cdot M=C$ still allows us to recover the updates of $V$, by observing that $V[((k-1)\alpha/|R|+1+q+1)|R|+j]=1$ if $x^{q}$ appears with non-zero coefficient in the polynomial at $C[k,j]$. 
This is because at row $A[k,\cdot]$ there are the coefficients that correspond to $U[\gamma\cdot |R|+i-1]$ for $\gamma=(k-1)\alpha/|R|+1,\ldots,k\alpha/|R|$ and $i=1,\ldots,|R|$, and hence a $x^{q}$ appearing at $C[k,j]$ is equivalent to a $x^{q+(k-1)\alpha/|R|+1}$ at $M''[j]$.

\medskip

We are now in a position to prove the following result for type 3 strings.

\begin{restatable}{theorem}{thmtypethree}\label{thm:type3}
An instance of the AP problem where all strings are of type 3 can be solved in $\ctO(m^{\omega-1}+N)$ time.
\end{restatable}

\begin{proof}
Recall that we consider strings $S$ of type $3$ with root $R$ and substrings of $P$ with root $R$ together. We first analyze the time to process a single group containing a number of substrings of $P$ of total length $m'$ and a number of strings $S\in{\cal S}$ of total length $N'$. 
Let us denote by $R^{\alpha_{i}}$ the $i$-th considered substring of $P$ and further define $\alpha=\sum_{i}\alpha_{i}=m'/|R|$.

If $\beta > \alpha/|R|$ we use the first method and spend $\cO(\alpha_{i} \cdot |C(\beta)|)$ time, where $C(\beta)$ is the set of $(i,j)$ for this specific $\beta$. The overall time used for all applications of the first method is:
\begin{eqnarray*}
\sum_{i}\cO(\alpha_{i} \cdot \sum_{\beta > \alpha/|R|} |C(\beta)|) &=& \cO(\alpha/ |R^{\alpha/|R|}|\cdot \sum_{\beta > \alpha/|R|} |C(\beta)|\cdot |R^{\alpha/|R|}|) \\
&=& \cO(\sum_{\beta > \alpha/|R|} |C(\beta)| \cdot |R^{\beta}|) = \cO(N'),
\end{eqnarray*}
using the fact that $\sum_{\beta}|C(\beta)| \cdot |R^{\beta}| \leq N'$ and $\alpha/ |R^{\alpha/|R|}|=\alpha/(\alpha/|R|)|R|=\cO(1)$.

For each $\alpha_i$, we process together all $\beta \le \alpha_i/|R|$ using the second method, and we need to multiply two $|R|\times |R|$ matrices of polynomials of degree up to $\alpha_i/|R|$, that we can build in time $\cO(N')$ and multiply in time $\cO(|R|^{\omega}\cdot\alpha_i/|R|+|R|^2(\alpha_i/|R|)\log (\alpha_i/|R|))$ by Lemma~\ref{lem:polymatrix}.
The overall time used for all applications of the second method is:

\begin{eqnarray*}
\cO(N')+\sum_{i}\cO(|R|^{\omega}\cdot\alpha_i/|R|+|R|^2(\alpha_i/|R|)\log (\alpha_i/|R|))=\cO(|R|^{\omega-2}m'+m'\log m'+N'),
\end{eqnarray*}

using the fact that $\alpha=m'/|R|$. Since $|R|\le m'$, this is in fact $\cO((m')^{\omega-1}+m'\log m'+N')$.

Because all values of $N'$ sum up to $N$ and all values of $m'$ sum up to $\cO(m)$, by convexity of $x^{\omega-1}$ we obtain that the overall time complexity is $\ctO(m^{\omega-1}+N)$.
\end{proof}

\subsection{Wrapping Up}
\label{ssec:wrap}

In Sections~\ref{sec:type1},~\ref{app:type2} and~\ref{app:type3} we design three $\ctO(m^{\omega-1}+N)$-time algorithms for an instance of the AP problem where all strings are of type 1, 2 and 3, respectively. We thus obtain Theorem~\ref{thm:algo}, and using the fact that $\omega<2.373$~\cite{DBLP:conf/issac/Gall14a,DBLP:conf/stoc/Williams12} we obtain the following corollary.

\begin{corollary}
The EDSM problem can be solved on-line in $\cO(nm^{1.373} + N)$ time.
\end{corollary}

Note that the polylog factors are shaved from $\ctO(nm^{\omega-1}+N)$ by using the fact that the inequality of $\omega<2.373$ is strict.

\section{Final Remarks}
Our contribution in this paper is twofold. First, we designed an appropriate reduction showing that a combinatorial algorithm solving the EDSM problem in $\cO(nm^{1.5-\epsilon} + N)$ time, for any $\epsilon>0$, refutes the well-known BMM conjecture. Second, we designed a non-combinatorial $\ctO(nm^{\omega-1}+N)$ -time algorithm to attack the same problem. By using the fact that $\omega<2.373$, our algorithm runs in $\cO(nm^{1.373} + N)$ time thus breaking the combinatorial conditional lower bound for the EDSM problem.
Let us point out that if $\omega=2$ then our algorithm for the AP problem is time-optimal up to polylog factors.

\section*{Acknowledgments}

GR and NP are partially supported by MIUR-SIR project CMACBioSeq ``Combinatorial methods for analysis and compression of biological sequences’’ grant n.~RBSI146R5L.

\vspace{0.2cm}
\begin{minipage}{0.1\textwidth}
\includegraphics[height=1cm]{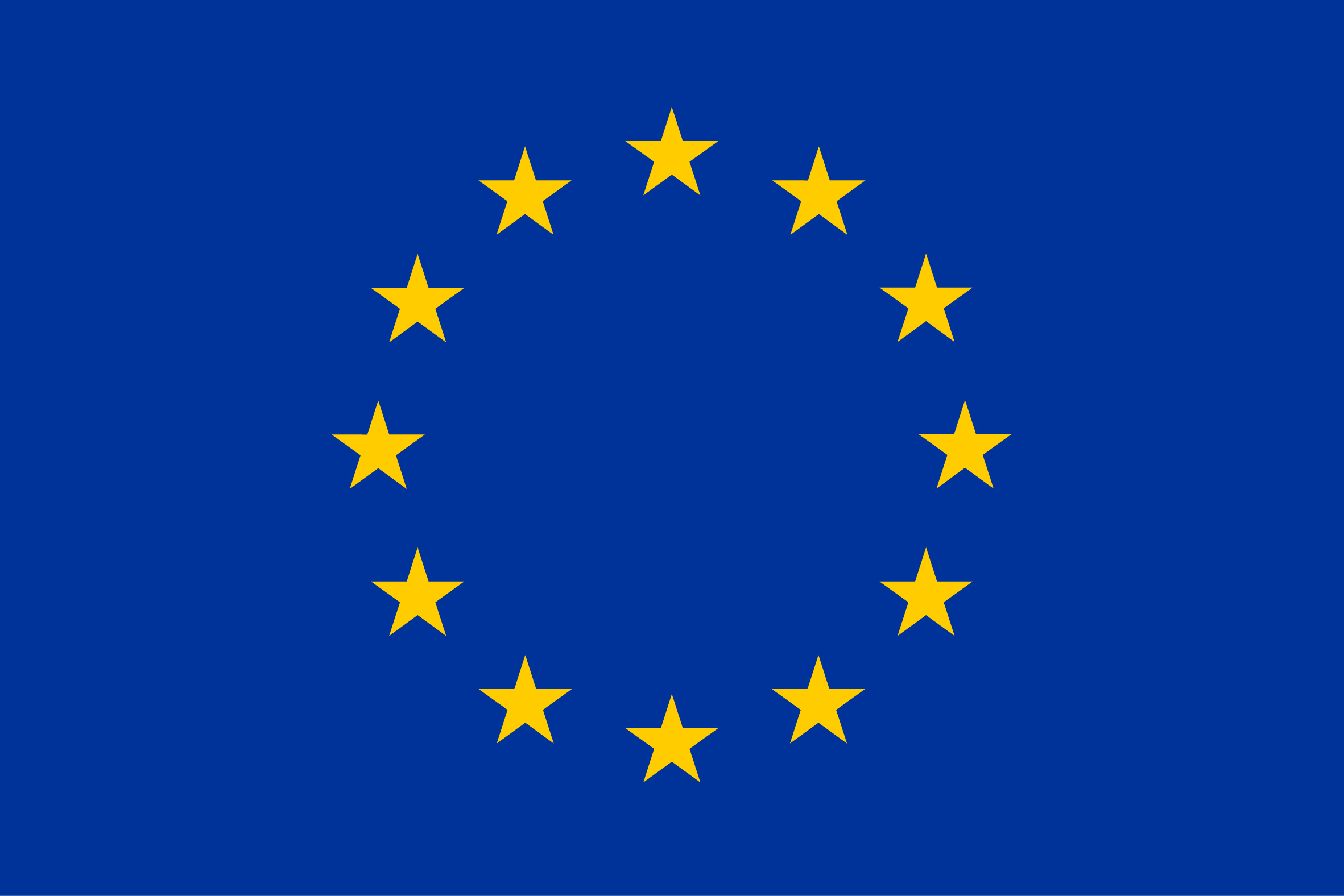}
\end{minipage}
\hfill
\begin{minipage}{0.85\textwidth}
This project has received funding from the European Union's
Horizon 2020 research and innovation programme under the
Marie Sk\l{}odowska-Curie grant agreement No 872539.
\end{minipage}

\bibliographystyle{plain}
\bibliography{references}

\end{document}